\documentclass[12pt]{amsart}
\usepackage{url}
\usepackage{amsmath,amsxtra,amssymb,latexsym,epsfig,amscd,amsthm,fancybox,epsfig}
\usepackage[mathscr]{eucal}
\usepackage{graphicx}
\usepackage{multicol,xcolor}
\usepackage{epsfig} % for postscript graphics files
\usepackage{epstopdf}
\usepackage{cases}
\usepackage{subfig}
\usepackage{color}
\usepackage{hyperref}
\usepackage[section]{placeins}

\hypersetup{
    colorlinks=true,
    linkcolor=blue,
    filecolor=magenta,
    urlcolor=cyan,
    citecolor=blue,
}

\usepackage{bigints}
\usepackage{enumitem}
\usepackage{natbib}

\newcommand{\Se}{\mathcal{S}}

\newtheorem{thm}{Theorem}[section]
\newtheorem {asp}{Assumption}[section]
\newtheorem{lm}{Lemma}[section]
\newtheorem{rmk}{Remark}[section]

\newtheorem{prop}{Proposition}[section]
\theoremstyle{definition}

\theoremstyle{remark}
\newtheorem{rem}{Remark}
\newtheorem{example}{Example}[section]
\numberwithin{equation}{section}

\DeclareMathOperator{\Conv}{Conv}
\newcommand{\eps}{\varepsilon}

\newcommand{\M}{\mathcal{M}}

\newcommand{\E}{\mathbb{E}}

\newcommand{\BX}{\mathbf{X}}
\newcommand{\bx}{\mathbf{x}}
\newcommand{\BY}{\mathbf{Y}}
\newcommand{\by}{\mathbf{y}}

\newcommand{\bu}{\mathbf{u}}

\newcommand{\bp}{\mathbf{p}}

\newcommand{\N}{\mathbb{N}}

\newcommand{\PP}{\mathbb{P}}

\newcommand{\R}{\mathbb{R}}
\newcommand{\Z}{\mathbb{Z}}

\newcommand{\U}{\mathcal{U}}

\numberwithin{equation}{section}

\newcommand{\bed}{\begin{displaymath}}
\newcommand{\eed}{\end{displaymath}}
\newcommand{\bea}{\bed\begin{array}{rl}}
\newcommand{\eea}{\end{array}\eed}

\newcommand{\barray}{\begin{array}{ll}}
\newcommand{\earray}{\end{array}}

\newcommand{\1}{\boldsymbol{1}}

\def\bar{\overline}
\def\hat{\widehat}
\def\a.s{\text{\;a.s.\;}}

\begin{document}

\title[Coexistence, extinction, and optimal harvesting]{Coexistence, extinction, and optimal harvesting in discrete-time stochastic population models}

\author[A. Hening]{Alexandru Hening }
\address{Department of Mathematics\\
Tufts University\\
Bromfield-Pearson Hall\\
503 Boston Avenue\\
Medford, MA 02155\\
United States
}
\email{alexandru.hening@tufts.edu}
\thanks{The author is being supported by the NSF through the grant DMS-1853463}

\maketitle

\begin{abstract}
We analyze the long term behavior of interacting populations which can be controlled through harvesting. The dynamics is assumed to be discrete in time and stochastic due to the effect of environmental fluctuations. We present powerful extinction and coexistence criteria when there are one or two interacting species. We then use these tools in order to see when harvesting leads to extinction or persistence of species, as well as what the optimal harvesting strategies, which maximize the expected long term yield, look like. For single species systems, we show under certain conditions that the optimal harvesting strategy is of bang-bang type: there is a threshold under which there is no harvesting, while everything above this threshold gets harvested. We are also able to show that stochastic environmental fluctuations will, in most cases, force the expected harvesting yield to be lower than the deterministic maximal sustainable yield.

The second part of the paper is concerned with the analysis of ecosystems that have two interacting species which can be harvested. In particular, we carefully study predator-prey and competitive Ricker models. We are able to analytically identify the regions in parameter space where the species coexist, one species persists and the other one goes extinct, as well as when there is bistability. We look at how one can find the optimal proportional harvesting strategy. If the system is of predator-prey type the optimal proportional harvesting strategy is, depending on the interaction parameters and the price of predators relative to prey, either to harvest the predator to extinction and maximize the asymptotic yield of the prey or to not harvest the prey and to maximize the asymptotic harvesting yield of the predators. If the system is competitive, in certain instances it is optimal to drive one species extinct and to harvest the other one.  In other cases it is best to let the two species coexist and harvest both species while maintaining coexistence.

In the setting of the competitive Ricker model we show that if one competitor is dominant and pushes the other species to extinction, the harvesting of the dominant species can lead to coexistence.
\end{abstract}

\maketitle
\tableofcontents

\section{Introduction}
A fundamental problem in population biology has been to find conditions for when interacting species coexist or go extinct. Since the dynamics of interacting populations is invariably influenced by the random fluctuations of the environment, realistic mathematical models need to take into account the joint effects of biotic interactions and environmental stochasticity. A successful way of analysing the persistence and extinction of interacting species has been to look at Markov processes, in either discrete or continuous time, and describe their asymptotic properties. There has been a recent resurgence in stochastic population dynamics, and significant progress has been made for stochastic differential equations (\cite{SBA11, HN18}), piecewise deterministic Markov processes (\cite{BL16, HS19, HN19}), stochastic difference equations (\cite{BS19}), and general Markov processes (\cite{B18}). The first focus of this paper is to present new results for persistence and extinction in the setting of stochastic difference equations when there are one or two interacting species. These results significantly generalize the work by \cite{CE89, E89} which only treated competition models and had no extinction results as well as the more recent work by \cite{BS19} which only looks at compact state spaces. We are able to give explicit conditions for extinction and persistence in the setting of competitive or predator-prey Ricker equations with random coefficients, adding to the previously known results \cite{E89, VH97, FH02, SBA11}. Our results involve computing the invasion rates (\cite{T78, C82, E84, CE89}) of each species into the random equilibrium of the other species. We show that if both invasion rates are strictly positive then there is coexistence. If instead, one invasion rate is positive and one is negative, the species with the positive invasion rate persists while the one with the negative invasion rate goes extinct. If there is coexistence we prove that, under natural conditions, the populations converge to a unique invariant probability measure. If there is extinction we show that, with probability one, one or both species go extinct exponentially fast. The general theory for the setting with $n>2$ interacting species will appear in future work by the author and his collaborators (\cite{CHN19}).

Once criteria for persistence and extinction are established, our focus shifts towards a key problem from conservation ecology: what is the optimal strategy for harvesting species? This is a delicate issue as overharvesting can sometimes lead to extinction while underharvesting can mean the loss of precious economic resources. In continuous time models, recent studies have been able to find the optimal harvesting strategy, which maximizes either the discounted total yield or the asymptotic yield under very general assumptions if the ecosystem has only one species (\cite{Alvarez98, HNUW18, AH19}). For multiple species the theory is less developed. Nevertheless, partial results exist (\cite{LO97, ALO16, Ky15, Ky16, Ky18, HT19}).

Quite often harvesting models are intrinsically discrete in time. For example, if one looks at the management of fisheries, most models (\cite{GH89, HW92, C10, HL19}) assume that the population in a given year can be described by a single continuous variable and that without harvesting the population levels in successive years are related by
\[
x_{n+1}=F(x_n)
\]
where $F$ is the so-called recruitment function or the reproduction function. Most discrete time harvesting results ignore random environmental fluctuations and their effects on the availability of food, competition rates, growth and death rates, strength of predation and other key factors. Ignoring environmental stochasticity can create significant problems, in some cases making the models unrealistic (\cite{May78}) and hard to fit to data (\cite{L73}).
A series of key studies where environmental fluctuations are included was done by \cite{R78, R79}. Reed looked at the setting where there is one species whose dynamics in the absence of harvesting is given by
\[
X_{n+1}=Z_n F(X_n)
\]
where $(Z_n)_{n\in\Z_+}$ is a sequence of i.i.d. random variables. We extend Reed's analysis in two ways. First, we study the more general stochastic difference equation
\[
X_{n+1}=F(X_n, \xi_{n+1})
\]
where $(\xi_n)_{n\in\Z_+}$ is an i.i.d sequence. Second, we are able to analyze systems of two interacting species. To our knowledge these are the first results in discrete time that study the harvesting of multiple species.
\subsection*{Single species ecosystem.}
We are able to give exact conditions under which harvesting leads to persistence or to extinction. In particular, we show that if there is only one species present, then the criteria for persistence only involve the harvesting rate of the population at $0$. We are able to find the maximal harvesting rate which does not lead to extinction. If the species $Y_t$ undergoing harvesting persists we prove it converges in law to a random variable $Y_\infty$ and, if the fraction of the population that gets harvested is given by the strategy $h(y)$, we show that the long run average and the expected long term harvest both converge to $\E h(Y_\infty)$. In many applications one is interested in seeing how the environmental fluctuations change the long term yield. We show that in most settings the environmental fluctuations are detrimental and lower the harvesting yield. Only in special cases, can we have that the maximum deterministic sustainable yield is equal to the steady state harvest yield of the stochastic system.

An interesting corollary of our results is that threshold harvesting strategies (also called constant-escapement strategies), where one does not harvest anything below a threshold and harvests everything above that threshold, do not influence the persistence of species as long as the threshold at which one starts harvesting is strictly positive. We showcase two examples where environmental fluctuations are not detrimental for threshold harvesting: 1) The threshold $w$ at which we harvest is self-sustaining, i.e., if at the start of the year we are at level $w$ the fluctuations of the environment can not push the population's size under $w$. In this setting the expected value of the long term yield in the stochastic model equals the yield from the equivalent deterministic model. The downside is that the variance of the yield is higher due to the environmental fluctuations. 2) The threshold $w$ is not self sustaining and the maximum yield of the dynamics happens at a self-sustaining threshold $\bar x<w$. In this case the expected yield of the constant escapement strategy is strictly greater than the yield of the same strategy in the deterministic system.

When looking at constant effort harvesting strategies, where a constant proportion of the return is captured every year, we show that even though the deterministic model might say that we harvest at a sustainable rate, the environmental fluctuations might lead to extinction.

We are able to say more in the setting of the Ricker model. We give conditions under which we can get the same maximal yield in the deterministic and stochastic settings. This includes giving information about the threshold for which the yield is maximized. We also find the optimal harvesting strategy if we restrict ourselves to proportional harvest strategies.
\subsection*{Two interacting species.}
We analyze a system of two interacting species that can be exploited through harvesting. We show that threshold harvesting strategies do influence the persistence criteria, unlike in the single species setting. In order to be able to compute things explicitly we focus on Ricker, also called discrete time Lotka-Volterra, models and assume that the harvesting strategy is of proportional type, where we harvest a fraction $q\in [0,1]$ of the first species and a fraction $r\in [0,1]$ of the second species.

The first studied model is a predator-prey system where species $1$ is the prey and species $2$ the predator. We give analytical expressions for when one has the persistence of both species, the persistence of the prey and the extinction of the predator as well as the extinction of both species. These expressions tell us exactly for which rates $q,r$ we get one of the three scenarios above. Which strategy, among all proportional harvesting strategies, maximizes the expected long term harvesting yield? We find by using both analytical results and numerical simulations that it is never optimal to harvest both the predator and the prey. Either we drive the predator extinct and we harvest the prey or we do not harvest the prey at all and we harvest the predator.

The second model we look at consists of an ecosystem where the two species compete with each other for resources. We  show that, depending on the inter and intracompetition coefficients of the system, one can have two different regimes each having three regions which depend on the harvesting rates $q, r$:
\begin{itemize}
  \item [a)] I) Persistence of species $1$ and extinction of species $2$; II) Extinction of species $1$ and persistence of species $2$; III) Coexistence
  \item [b)] I) Persistence of species $1$ and extinction of species $2$; II) Extinction of species $1$ and persistence of species $2$; III) Bistability
\end{itemize}
We show that harvesting can facilitate coexistence in certain cases. When species $1$ is dominant and drives species $2$ extinct in the absence of harvesting, it is possible to harvest species $1$ and ensure the persistence of both species.

Finally, we look at the optimal harvesting strategies for the competitive system. Combining analytical proofs and numerical simulations we see that, in contrast to the predator-prey setting, it can be optimal, depending on the inter and intra competition rates, to harvest one or both of the species.

\section{Stochastic population dynamics}
We start by describing the stochastic population models we will be working with. To include the effects of random environmental fluctuations, ecologists and mathematicians often use stochastic difference equations of the form
\begin{equation}\label{e:SDE}
X^i_{t+1} = X_t^i f_i(\BX_t, \xi_{t+1}).
\end{equation}
Here the vector $\BX_t:=(X_t^1,\dots,X_t^n)\in \Se\subset \R_+^n$ records the abundances of the $n$ populations at time $t\in\Z_+$ and $\xi_{t+1}$ is a random variable that describes the environmental conditions between time $t$ and $t+1$. The subset $\Se$ will denote the state space of the dynamics. It will either be a compact subset of $\R_+^n$ or all of $\R_+^n$.
The coexistence set is the subset $\Se_+=\{\bx\in \Se~|~x_i>0, i=1,\dots n\}$ of the state space where no species is extinct.
The real function $f_i(\BX_t,\xi_{t+1})$ captures the fitness of the $i$-th population at time $t$ and depends both on the population sizes and the environmental state. Models of this type can capture complex short-term life histories and include predation, cannibalism, competition, and seasonal variations.

We have to differentiate between the setting where the dynamics is bounded, and the process enters and remains in a compact set, and the case when the dynamics is unbounded.  We will make the following assumptions throughout the paper:
\begin{itemize}
\item[\textbf{(A1)}] $\xi_1,\dots,\xi_n,\dots$ is a sequence of i.i.d. random variables taking values in a Polish space $E$.
\item[\textbf{(A2)}] For each $i$ the fitness function $f_i(\bx,\xi)$ is continuous in $\bx$ on $\Se$, measurable in $(\bx,\xi)$ and strictly positive.
\end{itemize}

Assumptions (A1) and (A2) ensure that the process $\BX_t$ is a Feller process that lives on $\Se_+$, i.e. $\BX_t\in\Se_+, t\in \Z_+$ whenever $\BX_0\in\Se_+$. One has to make extra assumptions (see (A3) or (A4) in Appendix \ref{s:a1}) in order to ensure the process does not blow up or fluctuate too abruptly between $0$ and $\infty$. We note that most ecological models will satisfy these assumptions. For more details see the work by \cite{BS19, CHN19}.
\begin{rmk}\label{r:1}
Suppose the dynamics is given by the more general model of the type
\begin{equation}\label{e:SDE2}
X^i_{t+1} = F_i(\BX_t, \xi_{t+1}).
\end{equation}
Note that \eqref{e:SDE2} reduces to \eqref{e:SDE} if $F_i$ is $C^1$ and $F_i(\bx)=0$ whenever $x_i=0$.This means that $F_i$ is a nice, sufficiently smooth, vector field whick takes the value $0$ if species $i$ is extinct -- this is a natural assumption as there is no reason the population should be able to come back from extinction. Under these assumptions we can see that \eqref{e:SDE} is satisfied by setting
\begin{equation*}
   f_i(\bx,\xi) = \begin{cases}
              \frac{F_i(\bx,\xi)}{x_i} & \text{if } x_i>0,\\
               \frac{\partial F_i(\bx,\xi)}{\partial x_i} & \text{if } x_i = 0.
          \end{cases}
\end{equation*}
\end{rmk}

We will sometimes compare the stochastic model \eqref{e:SDE} with its averaged deterministic counterpart
\begin{equation}\label{e:det}
\bx_{t+1}^i=x_t^i \bar f_i(\bx_t)
\end{equation}
where $\bar f_i(\bx) := \E f_i(\bx,\xi_1)$. Note that
\[
\E[X^i_{t+1}~|~\BX_t=\bx] = x^i\E f_i(\bx,\xi_1) = x^i  \bar f_i(\bx),
\]
so that \eqref{e:det} is the average of \eqref{e:SDE} in this sense.

For example, if $f(x,\xi)=\xi u(x)$ and $\xi_1$ is a random variable with expectation $\E \xi_1=1$ then
\[
\E[X^i_{t+1}~|~\BX_t=\bx] = x^i u(\bx).
\]
This is the setting used by \cite{R78}.
\subsection{Stochastic persistence.}
We define the extinction set, where at least one species is extinct, by
\[
\Se_0:=\Se\setminus \Se_+=\{\bx\in\Se~:~\min_i x_i=0\}.
\]
The transition operator $P: \mathcal{B}\to \mathcal{B}$ of the process $\BX$ is an operator which acts on Borel functions $\mathcal{B}:=\{h:\Se\to\R~|~h~\text{Borel}\}$ as
\[
Ph(\bx)=\E_\bx[h(\BX(1))]:=\E[h(\BX(1))~|~\BX(0)=\bx],  ~\bx\in\Se.
\]
The operator $P$ acts by duality on Borel probability measures $\mu$ by $\mu\to \mu P$ where $\mu P$ is the probability measure given by
\[
\int_{\Se}h(\bx) (\mu P)(d\bx):=\int_{\Se}P h(\bx) \mu(d\bx)
\]
for all $h\in C(\Se)$.
A Borel probability measure $\mu$ on $\Se$ is called an \textit{invariant probability measure} if
\[
\mu P = \mu
\]
where $P$ is the transition operator of the Markov process $\BX_t$. An invariant probability measure or stationary distribution is a way of describing a `random equilibrium'. If the process starts with $\BX_0$ having an initial distribution given by the invariant probability measure $\mu$ then the distribution of $\BX_t$ is $\mu$ for all $t\in\Z_+$. In a sense this is the random analogue of a fixed point of a deterministic dynamical system. It turns out that a key concept is the \textit{realized per-capita growth rate} of species $i$ when introduced in the community described by an invariant probability measure $\mu$
\begin{equation}\label{e:r}
r_i(\mu) =\int_{\R_+^n}\E[\ln f_i(\bx,\xi_1)]\,\mu(d\bx) = \int r_i(\bx)\mu(d\bx)
\end{equation}
where
\[
r_i(\bx) = \E [\ln f_i(\bx,\xi_1)]
\]
is the mean per-capita growth rate of species $i$ at population state $\bx$. This quantity tells us whether species $i$ tends to increase or decrease when introduced at an infinitesimally small density into the the subcommunity described by $\mu$. If the $i$th species is among the ones supported by the subcommunity given by $\mu$, i.e., $i$ lies in the support of $\mu$, then this species is in a sense `at equilibrium' and one can prove that
\begin{equation}\label{e:equ}
r_i(\mu)=0.
\end{equation}
The only directions $i$ in which $r_i(\mu)$ can be non-zero are those which are not supported by $\mu$.

One can show that the invariant probability measures living on the extinction set $\Se_0$ fully describe the long term behavior of the system. In a sense, if any such invariant probability measure is a \textit{repeller} which pushes the process away from the boundary in at least one direction, then the system persists. Let $\Conv(\M)$ denote the set of all invariant probability measures supported on $\Se_0$. In order to have the convergence of the process to a unique stationary distribution one needs some irreducibility conditions which keep the process from being too degenerate \citep{CHN19, MT, B18}. The following theorem characterizes the coexistence of the ecosystem.
\begin{thm}\label{t:pers1_disc}
Suppose that for all $\mu\in\Conv(\M)$ we have
\begin{equation}\label{e:pers}
\max_{i} r_i(\mu)>0.
\end{equation}
Then the system is almost surely stochastically persistent and stochastically persistent in probability. Under additional irreducibility conditions, there exists a unique invariant probability measure $\pi$ on $\Se_+$ and as $t\to\infty$ the distribution of $\BX_t$ converges in total variation to $\pi$ whenever $\BX(0)=\bx\in\Se_+$. Furthermore, if $w:\Se_+\to\R$ is bounded
then $$\lim_{t\to\infty}\E w(\BX_t) = \int_{\Se_+} w(\bx)\,\pi(d\bx).$$
\end{thm}
A sketch of the proof of this result appears in Appendix \ref{s:a1}.
\subsection{Classification of two species dynamics}\label{s:2d}
Sometimes one is not only interested in persistence and coexistence, but also in conditions which lead to extinction. Extinction results are more delicate and require a technical analysis. Some extinction results appeared in work by \cite{HN18, BS19}. The sharpest and most complete results for stochastic difference equations and species feedbacks will appear in \cite{CHN19}. We restrict our discussion to a system with two species.
In this setting \eqref{e:SDE} becomes
\begin{equation}\label{e:2d_sys}
\begin{aligned}
X^1_{t+1}&=X^1_t f_1(X^1_t, X^2_t,\xi_{t+1}),\\
X^2_{t+1}&=X^2_t f_2(X^1_t, X^2_t,\xi_{t+1}).
\end{aligned}
\end{equation}
The exact assumptions and technical results are found in Appendix \ref{s:a2d}. We can classify the dynamics as follows. We first look at the Dirac delta measure $\delta_0$ at the origin $(0,0)$
\[
r_i(\delta_0) = \E[\ln f_i(0,\xi_1)], i=1,2.
\]
If $r_i(\delta_0)>0$ then species $i$ survives on its own and converges to a unique invariant probability measure $\mu_i$ supported on $\Se_+^i := \{\bx\in\Se~|~x_i\neq 0, x_j=0, i\neq j\}$. The realized per-capita growth rates can be computed as
  \[
  r_i(\mu_j)=\int_{(0,\infty)}\E[\ln f_i(x,\xi_1)]\mu_j(dx).
  \]
\begin{enumerate}[label=(\roman*)]
  \item Suppose $r_1(\delta_0)>0, r_2(\delta_0)>0$.
  \begin{itemize}
    \item If $r_1(\mu_2)>0$ and $r_2(\mu_1)>0$ we have coexistence and convergence of the distribution of $\BX_t$ to the unique invariant probability measure $\pi$ on $\Se_+$.
    \item If $r_1(\mu_2)>0$ and $r_2(\mu_1)<0$ we have the persistence of $X^1$ and extinction of $X^2$.
    \item If $r_1(\mu_2)<0$ and $r_2(\mu_1)>0$ we have the persistence of $X^2$ and extinction of $X^1$.
    \item If $r_1(\mu_2)<0$ and $r_2(\mu_1)<0$ we have that for any $\BX_0=\bx\in\Se_+$
    \[
    p_{\bx,1}+ p_{\bx,2}=1,
    \]
    where $p_{\bx,j}>0$ is the probability that species $j$ persists and species $i\neq j$ goes extinct.
  \end{itemize}
  \item Suppose $r_1(\delta_0)>0, r_2(\delta_0)<0$. Then species $1$ survives on its own and converges to its unique invariant probability measure $\mu_1$ on $\Se^1_+$.
  \begin{itemize}
    \item If $r_2(\mu_1)>0$ we have the persistence of both species and convergence of the distribution of $\BX_t$ to the unique invariant probability measure $\pi$ on $\Se_+$.
    \item If $r_2(\mu_1)<0$ we have the persistence of $X^1$ and the extinction of $X^2$.
  \end{itemize}
\item Suppose $r_1(\delta_0)<0, r_2(\delta_0)<0$. Then both species go extinct with probability one.
\end{enumerate}
We note that our results are significantly more general than those from \cite{E89}. In \cite{E89} the author only gives conditions for coexistence, and does not treat the possibility of the extinction of one or both species.

\begin{example}\label{ex:1}
The simplest case is when the noise is multiplicative, that is
\begin{equation}\label{e:2d_syss}
\begin{aligned}
X^1_{t+1}&=X^1_t Z_{t+1}^1f_1(X^1_t, X^2_t)\\
X^2_{t+1}&=X^2_t Z_{t+1}^2 f_2(X^1_t, X^2_t),
\end{aligned}
\end{equation}
where $Z^1_{1}, Z^1_{2},\dots$ is an i.i.d sequence of random variables and $Z^2_{1}, Z^2_{2}, \dots$ is an independent sequence of i.i.d random variables.
In this case for $i=1,2$ we have
\begin{equation}\label{e:2d_pers}
\begin{aligned}
r_i(\delta_0) &= \E[\ln (Z_{t+1}^i f_i(0))]\\
&= \E \ln Z_1 + \ln f_i(0).
\end{aligned}
\end{equation}
The growth rates at $0$ in the stochastic model differ from the growth rates at $0$ of the deterministic model only by the term $\E \ln Z_1$.
\end{example}

\subsection{Harvesting}\label{s:harvesting}
We next describe how the harvesting effects are taken into account. We assume that the harvesting takes place during a short harvest season. The size of the population at the beginning of the harvest season in year $t$ will be denoted by $\BY_t$ and will be called \textit{return in year $t$}. If we assume the harvest season is short so that growth and natural mortality can be neglected during the harvesting and that the harvesting strategy is \textit{stationary}, i.e., the size of the harvest in any year depends only on the size of the population return $\BY$ in that year we can write
\begin{equation}\label{e:harv1}
X^i_t = Y_t^i - h_i(\BY_t) = u_i(\BY_t)
\end{equation}
where $X^i_t$ is the escapement of the $i$th population from the harvest and $h_i(\BY_t)$ is the amount of species $i$ that is harvested at time $t$. The function $u_i$ is called the \textit{escapement function} and measures how much is left after harvesting. Note that, since we cannot harvest a negative amount or more than the total population size we will always have
\[
0\leq h_i(\by)\leq y_i.
\]
 Set $\bu(\by):=(u_1(\by),\dots,u_n(\by))$. Once the harvesting is done, the population evolves according to \eqref{e:SDE} so that the size of the return in year $t+1$ is related to the escapement in year $t$ via
\begin{equation}\label{e:harv2}
Y^i_{t+1} = X_t^i f_i(\BX_t, \xi_{t+1}).
\end{equation}
Combining \eqref{e:harv1} and \eqref{e:harv2} we get
\begin{equation}\label{e:harv}
Y^i_{t+1} = u_i(\BY_t)  f_i(\bu(\BY_t), \xi_{t+1}).
\end{equation}
In order to be able to analyze the process $\BY_t$ we have to make sure that it can be written in the \textit{Kolmogorov} form \eqref{e:SDE}. In order to get this we assume that
\begin{asp}\label{a:1}
For all $i=1,\dots,n$ the following properties hold
\begin{itemize}
\item[(a)] The function $u_i$ is strictly positive on $\Se_+$, with
\[
u_i(\by) \leq y_i.
\]
\item [(b)] The function $u_i$ is continuous on $\Se$ and continuously differentiable at $y_i=0$.
\end{itemize}
\end{asp}
\begin{rem}\label{r:2}
Note that Assumption \ref{a:1} implies that $u_i(\by)=0$ if $y_i=0$ and
\[
\frac{\partial u_i}{\partial y_i}(\by)=\lim_{y_i\to 0} \frac{u_i(\by)}{y_i}\leq 1
\]
if $\by\in\Se$ with $y_i=0$.
\end{rem}

\subsection{Persistence with harvesting.} Since overharvesting can lead to extinction, we want to find sufficient conditions which ensure the process $\BY_t$ converges to a unique invariant probability measure on $\Se$.  Note that we need to put \eqref{e:harv} into the form \eqref{e:SDE}. For this, using Remark \ref{r:1}, let
\begin{equation}\label{e:g}
   g_i(\bx,\xi) := \begin{cases}
              \frac{u_i(\bx)}{x_i} f_i(\bu(\bx), \xi)& \text{if } x_i>0,\\
               \left(\frac{\partial u_i}{\partial x_i}(\bx)\right)f_i(\bu(\bx), \xi) & \text{if } x_i = 0.
          \end{cases}
\end{equation}
We can write \eqref{e:harv} as
\begin{equation}\label{e:harv3}
Y^i_{t+1} = Y^i_t  g_i(\BY_t, \xi_{t+1}).
\end{equation}
In order to use Theorem \ref{t:pers1_disc} we have to make sure that conditions A1)-A4) are satisfied, that the process $\BY_t$ is $\phi$-irreducible and that \eqref{e:pers} holds.
If $\mu$ is an invariant probability measure of $\BY_t$ living on the extinction set $\Se_+$ the realized per capita growth rates will be given by
\begin{equation}\label{e:r2}
r_i(\mu) =\int_{ \Se}\E[\ln g_i(\bx,\xi_1)]\,\mu(d\bx).
\end{equation}
Specifically, if we look at the Dirac mass at $0$ we get
\begin{equation}\label{e:nd_pers3}
\begin{aligned}
r_i(\delta_0)&=\int_{\Se_0}\E[\ln g_i(\bx,\xi_1)]\,\delta_0(d\bx)\\
&= \E[\ln g_i(0,\xi_1)] \\
&= \E\left[\ln \left(\frac{\partial u_i}{\partial x_i}(0)f_i(0, \xi)\right)\right].
\end{aligned}
\end{equation}
\begin{example}
If the noise is multiplicative and we are in the setting of Example \ref{ex:1}, i.e.,
\begin{equation}\label{e:nd_sys}
\begin{aligned}
Y^i_{t+1}&=Y^i_t Z_{t+1}^ig_i(\BY_t)
\end{aligned}
\end{equation}
then
\begin{equation}\label{e:nd_pers2}
\begin{aligned}
r_i(\delta_0) &= \E[\ln (Z_{t+1}^i g_i(0))]\\
&= \E \ln Z_1 +\ln \left(\frac{\partial u_i}{\partial x_i}(0)f_i(0)\right)\\
&= \E \ln Z_1 +\ln \left(\frac{\partial u_i}{\partial x_i}(0)\right) + \ln f_i(0).
\end{aligned}
\end{equation}
\textit{\textbf{Biological interpretation:} The noise can either help or inhibit the survival of the $i$th species, according to whether $\E \ln Z_1<0$ or $\E \ln Z_1>0$. Since $\frac{\partial u_i}{\partial x_i}(0)\leq 1$ we always have $\ln \left(\frac{\partial u_i}{\partial x_i}(0)\right)\leq 0$ so that, as expected, harvesting is always detrimental to the survival of each individual species. The minimal escapement rate under which one can still have persistence is such that
\[
\frac{\partial u^{min}_i}{\partial x_i}(0) > e^{- (\E \ln Z_1 + \ln f_i(0))} = f_i(0)  e^{- \E \ln Z_1}.
\]
Since $h_i(\by)=1-u_i(\by)$ this implies that the maximal harvesting rate satisfies
\[
\frac{\partial h^{max}_i}{\partial x_i}(0) < 1- e^{- (\E \ln Z_1 + \ln f_i(0))}=1- f_i(0)  e^{- \E \ln Z_1}.
\]
}
\end{example}

\section{Single species harvesting}\label{s:1s}
This section explores the setting when there is only one species in the ecosystem. The results can be seen as an extension of the results from \cite{R78}. Our results show that environmental fluctuations are usually detrimental to the optimal harvesting yield. Actually, only under very special conditions, is it possible for the stochastic dynamics to have the same maximal expected long term yield as the related deterministic dynamics. Even in that case, the nonzero variance of the stochastic long term yield, makes it more risky than its deterministic analogue.

We can show that in some special cases of constant-escapement strategies, it is possible for the stochastic expected long term yield to be higher than the deterministic yield.

One can see from \eqref{e:harv3} that the dynamics of the return will be given by
\begin{equation}\label{e:harv4}
Y_{t+1} = Y_t  g(Y_t, \xi_{t+1})
\end{equation}
for
\begin{equation*}
   g(x) := \begin{cases}
              \frac{u(x)}{x} f(u(x), \xi)& \text{if } x>0,\\
               \left(\frac{\partial u}{\partial x}(0)\right)f(u(0), \xi) & \text{if } x = 0.
          \end{cases}
\end{equation*}
We will work under the assumption that without harvesting we have
\begin{equation}\label{e:r_single_pers}
\E\left[\ln \left(f(0, \xi_1)\right)\right]>0
\end{equation}
so that the species persists.
Suppose the assumptions of one of the Theorems \ref{t:pers_single1}, \ref{t:pers_single2}, \ref{t:pers_single3} or \ref{t:pers1_disc} hold. Then, in order to have persistence we need
\begin{equation}\label{e:r_single2}
r_1(\delta_0) =\int_{\partial \R_+}\E[\ln g(x,\xi_1)]\,\delta_0(dx) = \E\left[\ln \left(\left(\frac{\partial u}{\partial x}(0)\right)f(0, \xi_1)\right)\right]>0
\end{equation}
where $\delta_0$ is the point mass at $0$ and we made use of \eqref{e:r2} and Assumption \ref{a:1}.
We can express this result as
\begin{equation}\label{e:r_single}
\frac{\partial u}{\partial x}(0)> e^{-  \E\ln f(0, \xi_1)}.
\end{equation}
Let us next compute the expected long term harvest yield. If the assumptions of Theorem \ref{t:pers1_disc} are satisfied we will have
\[
Y_{t}\to Y_\infty
\]
in distribution as $t\to\infty$. Here $Y_\infty$ is a random variable whose distribution equal to the the invariant probability measure $\pi_u$. In many models, and for well behaved functions $h$ one can show by Theorem \ref{t:pers1_disc} that $\E h(Y_\infty)$ exists and is finite. As a result, we have that with probability one
\[
\lim_{T\to \infty} \frac{\sum_{t=0}^Th(Y_t)}{T} =  \E h(Y_\infty).
\]
This tells us that the \textit{long-run average harvest yield} converges to a steady yield $\E h(Y_\infty)$. Furthermore, we can also see that \textit{expected yield} also converges to the same quantity
\[
\lim_{T\to \infty} \E h(Y_T) = \E h(Y_\infty) = \int_{(0,\infty)}h(x)\,\pi_u(dx).
\]
From now on we will call $\E h(Y_\infty)$ the \textit{expected steady-state yield}.
In general it is not possible to find $\E h(Y_\infty)$. However, in certain instances we can exploit the fact that, at stationarity, the realized per-capita growth rates in the directions supported by the measure $\pi_u$ are all zero \citep{CHN19}. In other words
\begin{equation}\label{e:r0}
0=r_1(\pi) = \int_{(0,\infty)}\E\left[\ln\left( \frac{u(x)}{x} f(u(x), \xi_1)\right)\right]\,\pi_u(dx).
\end{equation}
\subsection{Stochastic versus deterministic harvesting.}\label{s:stocdet}
Let us compare the stochastic dynamics \eqref{e:harv4} with its deterministic average
\begin{equation}\label{e:det1}
x_{t+1}=x_t \bar f(x_t)
\end{equation}
where $\bar f(x) := \E f(x,\xi_1)$ and $\bar F(x) := x \bar f(x)$. If $h$ is any stationary harvesting strategy the deterministic equilibrium return $y$ satisfies
\[
y=\bar F(u(y))
\]
and the equilibrium yield is
\[
h(y)=y-u(y) = \bar F(u(y))-u(y)=\bar G (u(y))
\]
where $\bar G(x):=\bar F(x)-x$. The \textit{deterministic maximum sustainable yield} (DMSY) is obtained by keeping the escapement $u(y)$ at the level $x_1$ at which $\bar G$ attains its maximum, i.e., at the point $x_1$ where
\[
0=\bar G'(x_1)=\bar F'(x_1)-1=\bar f(x_1) + x_1 \bar f'(x_1)-1.
\]
The DMSY $M_{\det}$ will be
\[
M_{\det} = \bar G(x_1)=x_1\bar f(x_1)-x_1.
\]

\begin{thm}\label{t:maximal_yield}
The expected value of the steady state harvest yield $\E h(Y_\infty)$ of any stationary harvesting policy $h$ of the model \eqref{e:harv4} is always dominated by the maximum deterministic sustainable yield of the equivalent deterministic model \eqref{e:det1},
\[
\E h(Y_\infty) \leq M_{\det}.
\]
 The only way to achieve an equality in the above is when the following conditions are satisfied:
\begin{itemize}
  \item [1)] The unharvested dynamics $X_{t+1}=X_tf(X_t,\xi_{t+1})$ is able to go to a level greater or equal to $x_1$.
  \item [2)] The harvesting policy is bang-bang with threshold $x_1$, that is
  \begin{equation*}
   h^*(y) := \begin{cases}
             y-x_1 & \text{if } y>x_1,\\
               0 & \text{if } y\leq x_1.
          \end{cases}
\end{equation*}
  \item [3)] The level $x_1$ is self-sustaining i.e. the stochastic effects do not make the population ever go below $x_1$ once it reaches this level.
\end{itemize}
\end{thm}
\begin{proof}
Let $G(x,\xi):= x f(x,\xi)-x$ and note that $\E G(x,\xi_1)= \bar G(x)$. For the stochastic model if we use the harvesting policy $h$, the long run average yield is
\begin{equation}\label{e:hY}
\begin{split}
  \E h(Y_\infty) &=  \E Y_\infty - \E[u(Y_\infty)] \\
   &= \E u(Y_\infty) f(u(Y_\infty),\xi_1) - \E[u(Y_\infty)] \\
   &= \E G(u(Y_\infty)),\xi_1)\\
   &= \int \int G(u(y), \xi) \PP(Y_\infty\in dy, \xi_1\in d\xi)\\
   &= \int \int G(u(y), \xi) \PP(Y_\infty\in dy)\PP(\xi_1\in d\xi)\\
   &= \int \left(\int G(u(y), \xi)\PP(\xi_1\in d\xi)\right) \PP(Y_\infty\in dy)\\
   &= \E \bar G (u(Y_\infty))
\end{split}
\end{equation}
where we used the fact that $u(Y_\infty)$ and $u(Y_\infty) f(u(Y_\infty),\xi_1)$ have the same distribution and $Y_\infty$ is independent of $\xi_1$. Since $\bar G$ attains its maximum at $x_1$ we have
\begin{equation}\label{e:hYin}
 \E h(Y_\infty) = \E \bar G (u(Y_\infty))\leq \bar G(x_1)= M_{\det}.
\end{equation}
In order to have equality in \eqref{e:hYin} we need the law of $u(Y_\infty)$ to be the point mass $\delta_{x_1}$ at $x_1$. This means that with probability 1
\[
Y_\infty - h(Y_\infty) = x_1.
\]
One can achieve this if:
\begin{itemize}
\item [1)] the population can get to a level that is equal or greater to $x_1$,
\item [2)] one uses the \textit{bang-bang}, also called \textit{constant escapement} or \textit{threshold}, harvest policy at the level $x_1$
\begin{equation*}
   h^*(y) := \begin{cases}
             y-x_1 & \text{if } y>x_1,\\
               0 & \text{if } y\leq x_1,
          \end{cases}
\end{equation*}
and
\item [3)] once the population reaches the level $x_1$, it never decreases to a lower level, that is, if $X_t=x_1$ then
\[
Y_{t+1} = X_1f(X_1,\xi_1)=x_1f(x_1,\xi_1) \geq x_1
\]
with probability $1$.
\end{itemize}
 The last property is equivalent to having $\PP(f(x_1,\xi_1)\geq 1)=1$ -- if this is true, we say that the level $x_1$ is \textit{self-sustaining}.
\end{proof}
\textit{\textbf{Biological interpretation:} In general, it is not possible to have the same maximal yield in the stochastic setting as in the deterministic setting. Due to environmental fluctuations the expected long term yield of any harvesting strategy $h$ will be dominated by the deterministic maximum sustainable yield. The only case when the maximal yields in the stochastic and deterministic setting are equal, is when one uses a constant escapement strategy with threshold $x_1$ (which maximizes the deterministic MSY), the stochastic dynamics can reach levels greater or equal to $x_1$ and $x_1$ and then never goes below $x_1$ due to environmental fluctuations. These very specific conditions will not usually hold. As such, for most situations we cannot expect to get the same optimal harvesting yields in the deterministic and stochastic settings.
}

We note that threshold (or bang-bang) harvesting strategies do not influence the persistence criterion in the one dimensional case. The unharvested system
\begin{equation}\label{e:SDE_1d}
X_{t+1} = X_t f(X_t, \xi_{t+1})
\end{equation}
has
\[
r^X_1(\delta_0) = \E \ln f(0,\xi_1) >0.
\]
If one adds harvesting, then
\[
Y_{t+1} = Y_t  g(Y_t, \xi_{t+1})
\]
where
\begin{equation*}
   g(x) := \begin{cases}
              \frac{u_w(x)}{x} f(u_w(x), \xi)& \text{if } x>0,\\
               \left(\frac{\partial u_w}{\partial x}(0)\right)f(u_w(0), \xi) & \text{if } x = 0.
          \end{cases}
\end{equation*}
The bang-bang strategy
\begin{equation*}
   h_w(y) := \begin{cases}
             y-w & \text{if } y>w,\\
               0 & \text{if } y\leq w,
          \end{cases}
\end{equation*}
with $w>0$ also has
\[
r^Y_1(\delta_0) = \E \ln g(0,\xi_1)= \E \ln \left(\frac{\partial u_w}{\partial x}(0)\right)f(u(0), \xi) = \E \ln f(0,\xi_1)=r^X_1(\delta_0).
\]
\textit{\textbf{Biological interpretation:} This implies that a constant escapement strategy with a threshold $w>0$ does not change the per-capita growth rate and thus does not interfere with persistence. This is one reason why bang-bang harvesting strategies are robust and make sense when there is only one species present. This is not the case anymore when there are multiple species present.}

\subsection{Constant effort harvesting}
Quite often in fisheries a constant effort harvesting method is used. These strategies are such that the same fixed proportion of the return is captured every year. In other words, for some fixed $\theta \in (0,1)$ we have
\[
h_\theta(x)= \theta x
\]
and
\[
u_\theta(x)=(1-\theta)x.
\]
The persistence criteria \eqref{e:r_single} becomes
\[
\theta < \theta_{\max} := 1 - e^{-\E\ln f(0,\xi_1)}
\]
where $\theta_{\max}$ is the \textit{maximum sustainable rate of exploitation}.
Let us compare this with the deterministic system
\[
x_{t+1}=x_t \bar f(x_t) =\bar F(x_t)
\]
where $\bar f(x) = \E f(x,\xi_1)$. In this setting the maximum sustainable rate of exploitation $\theta_{\det}$ is given by
\[
\theta_{\det} = 1 - \frac{1}{\bar f(0)}=1- e^{-\ln \E f(0,\xi_1) }.
\]
Since the logarithm is a concave function, Jensen's inequality implies that
\[
 \E \ln f(0,\xi_1) \leq \ln \E f(0,\xi_1),
\]
with equality if and only if $f(0,\xi_1)$ is constant with probability one. As a result
\[
\theta_{\max}<\theta_{\det},
\]
which was shown by \cite{R78} in a simpler model.

\textit{\textbf{Biological interpretation:} This inequality is important because it says that if one uses a deterministic model, then the rate of exploitation might seem to be sustainable. Nevertheless, using this strategy will lead to extinction due to the presence of random environmental fluctuations.}

It is well known that in the setting of \eqref{e:det1} the deterministic maximum sustainable yield (DMSY) is achieved when the rate of exploitation is
\[
\theta_{DMSY} = 1 - \frac{1}{\bar f(x_1)}
\]
for $x_1$ satisfying $$\bar f'(x_1)=1.$$ Under environmental conditions which are large enough we can have
\[
e^{-\ln \E f(0,\xi_1) } > \frac{1}{\bar f(x_1)}.
\]
\textit{\textbf{Biological interpretation:} If the environmental fluctuations are significant, one has $\theta_{DMSY}>\theta_{\max}$. This shows that if large environmental fluctuations are possible and we harvest the population according to the deterministic MSY rate of exploitation we will drive it to extinction.}
\begin{thm}\label{t:CE}
If the deterministic averaged system \eqref{e:det1} has a strictly concave $\bar F$, and the dynamics \eqref{e:SDE_1d} is not purely deterministic, then the asymptotic expected yield of any constant effort harvest strategy is strictly lower than the deterministic yield of that harvesting policy.
\end{thm}
\begin{proof}
For a constant-effort policy $h_\theta(y)=\theta y$ the asymptotic expected yield is given by
\[
\E h_\theta(Y_\infty)=\theta \E Y_\infty.
\]
Set $u_\theta(y)=(1-\theta)y$. Using that $ u_\theta(Y_\infty) f(u_\theta(Y_\infty),\xi_1)$ and $u_\theta(Y_\infty)$ have the same distribution an argument similar to the one from \eqref{e:hY} shows that
\[
\E Y_\infty = \E Y_\infty f(u_\theta(Y_\infty),\xi_1) = \frac{1}{1-\theta}\E \bar F(u_\theta(Y_\infty)).
\]
If the function
\[
\bar F (y) = y \E f(y,\xi_1)
\]
is strictly concave then by Jensen's inequality
\begin{equation}\label{e:ineq}
 \frac{1}{1-\theta}\E \bar F(u_\theta(Y_\infty)) \leq  \frac{1}{1-\theta} \bar F(\E u_\theta(Y_\infty)).
\end{equation}
One can have equality in \eqref{e:ineq} only if $Y_\infty$ is with probability one a constant random variable. This implies that
\[
\E Y_\infty \leq \bar F ((1-\theta) \E Y_\infty).
\]
As we know, in the deterministic model, using the same policy with harvest rate $\theta$, the equilibrium return $\hat y_\theta$ satisfies
\[
\hat y_\theta = \bar F((1-\theta)\hat y_\theta).
\]
This together with the strict concavity of $\bar F$ implies that
\[
\E Y_\infty \leq \hat y_\theta.
\]
Equality can only hold if $Y_\infty$ is with probability one a constant, which means the dynamics is deterministic.
\end{proof}

\subsection{Bang-bang threshold harvesting.} Bang-bang or constant-escapement harvesting strategies are important and are used in many theoretical models as well as in actual harvesting situations, like fisheries. These policies have been shown to be optimal in many instances both for the continuous (\cite{LO97,Alvarez98, HNUW18, AH19}) and discrete time (\cite{R78, R79}) settings. Constant escapement strategies turned out to be optimal for maximizing discounted yield, asymptotic yield, as well as discounted economic revenue under many different conditions. In discrete time the work by \cite{R79} implies that a bang-bang policy maximizes the expected discounted net revenue in a discrete time stochastic model. It has not been shown in discrete time, to our knowledge, in generality that the expected steady state yield is maximized under a bang-bang strategy. However, both heuristic arguments and analytical results in specific cases hint that these strategies are probably the ones that will in general be optimal. In addition, these are the strategies that are most widely used in fisheries, where the escapement is controlled. We will explore how well these bang-bang strategies do in the stochastic harvesting setting \eqref{e:harv4} in comparison to the deterministic setting \eqref{e:det1}.

Suppose we harvest according to the bang-bang strategy
\begin{equation}\label{e:bang}
   h_w(y) := \begin{cases}
             y-w & \text{if } y>w,\\
               0 & \text{if } y\leq w,
          \end{cases}
\end{equation}
with $w>0$.
Let $r_s$ be the maximum self-sustaining level
\[
r_s = \max \{x~|~\PP(f(x,\xi_1)\geq 1)=1\}.
\]
We have to differentiate between two cases:
\begin{itemize}
  \item[1)] The level $w$ is self-sustaining, i.e.
  \[
  \PP(f(w,\xi_1)\geq 1)=1.
  \]
  \item[2)] The level $w$ is not self-sustaining.
\end{itemize}
\begin{prop}
If the threshold level $w$ is self-sustaining then the expected value of the long term yield $\E h_w(Y_\infty)$ is equal to the deterministic yield of the same strategy $\bar G (w)$. The variance of the yield $h_w(Y_\infty)$ is given by
\[
\sigma^2(h_w(Y_\infty)):=w^2 \E[f^2(w,\xi_1)- \bar f(w)^2].
\]
\end{prop}
\begin{proof}
Suppose $w$ is self-sustaining. Then
\[
\E h_w(Y_\infty)=\E(Y_\infty-w) = \E w f(w,\xi_1)-w=\bar F(w)-w=\bar G(w).
\]
For the variance of the yield we get
\begin{equation*}
\begin{split}
  \sigma^2(h_w(Y_\infty))&= \E[h_w(Y_\infty)^2] - (\E h(Y_\infty))^2\\
  &= \E [(wf(w,\xi_1)-w)^2] - (\bar F(w)-w)^2\\
  &=\E w^2 f^2(w,\xi_1) -2 w^2 \E f(w,\xi_1) + w^2 -\bar F^2(w) + 2\bar F(w) w - w^2\\
  &=w^2 \E[f^2(w,\xi_1)-(\E f(w,\xi_1))^2].
\end{split}
\end{equation*}
This completes the proof.
\end{proof}
\textit{\textbf{Biological interpretation:} Self-sustaining thresholds $w$, where the fluctuations never push the population's size under this threshold, make it possible to have the same yield in the stochastic and deterministic settings. However, the environmental fluctuations make the variance of the yield increase.
}
\begin{prop}
Suppose the following properties hold:
\begin{itemize}
  \item $w$ is not self-sustaining
  \item all the levels $x\in [0,r_s]$ are self-sustaining
  \item $\bar G(x)=x\bar F(x) - x$ is unimodal with its maximum at $\bar x$
  \item $\bar x$ is self-sustaining.
\end{itemize}
The expected steady state yield of the harvesting strategy $h_w$ is strictly greater than the deterministic nominal yield of the same strategy
\[
\E h_w(Y_\infty)>\bar G(w).
\]
\end{prop}
\begin{proof}
If $h_w$ is given by \eqref{e:bang} for some $w>0$ then
\begin{equation}\label{e:bang2}
   u_w(y) := \begin{cases}
             w & \text{if } y>w,\\
               y & \text{if } y\leq w.
          \end{cases}
\end{equation}
Note that $Y_\infty$ will be supported by a subset of $[r_s,\infty)$. By assumption $w>r_s>\bar x$, so that
\[
\bar G(y) > \bar G(w), y\in [r_s,w].
\]
This implies that with probability one
\[
\bar G(Y_\infty) \mathbf{1}\{Y_\infty\in [r_s,w]\}>\bar G(w)\mathbf{1}\{Y_\infty\in [r_s,w]\}.
\]
Using that the function $\bar G$ is nonincreasing on $[r_s,w]$ together with \eqref{e:bang2} and the last inequality we see that
\begin{equation*}
\begin{split}
\E h_w(Y_\infty) &= \E \bar G(u_w(Y_\infty))\\
&=\E [\bar G(Y_\infty) \mathbf{1}\{Y_\infty\in [r_s,w]\} ] + \E[\bar G(w) \mathbf{1}\{Y_\infty> w\}]\\
&>  \E [\bar G(w) \mathbf{1}\{Y_\infty\in [r_s,w]\} ] + \E[\bar G(w) \mathbf{1}\{Y_\infty> w\}]\\
&= \bar G(w).
\end{split}
\end{equation*}
with $w>0$.
\end{proof}
\textit{\textbf{Biological interpretation:} Suppose one picks a harvesting threshold $w$ which is not self-sustaining while the maximum yield of the deterministic dynamics happens at a threshold $\bar x<w$ which is self-sustaining. Then the expected yield of the constant escapement strategy with threshold $w$ for the stochastic dynamics is strictly greater than the expected yield of the same strategy in the deterministic system. The environmental fluctuations will push the population size into the region $(\bar x, w)$ and in this region, the function $\bar G$, which measures the size of the deterministic harvest, is strictly decreasing. This makes it more favorable to go below $w$, something which is not possible in the deterministic dynamics.
}

\section{The Ricker model: single species}
In this section we will provide an in depth analysis of the \textit{Ricker model}.
Its dynamics is given by the functional response
\[
f(x,\xi) = e^{\rho-\alpha x}.
\]
Here the randomness comes from $\xi:=(\rho,\alpha)$. The quantity $\rho_t$ is the fluctuating growth rate and $\alpha_t$ is the competition rate. We assume that $\rho_1, \rho_2,\dots$ are i.i.d random variables on $\R$, and $\alpha_1,\dots$ are independent i.i.d random variables supported on $\R$.
In this setting, one can see that without harvesting
\[
r^X(\delta_0) = \E \rho_1
\]
while with harvesting strategy $h(y)$ (or escapement strategy $u(y)$)
\[
r^Y(\delta_0)= \ln\left(\frac{\partial u}{\partial x}(0) h\right) + \E \rho_1.
\]
The maximal harvesting rate at $0$ which does not lead to extinction is
\[
\frac{\partial h}{\partial x}(0)<1- e^{-  \E\rho_1}.
\]
\subsection{Maximum sustainable yield.} We want to see when we can apply the results of Theorem \ref{t:maximal_yield}. Suppose that $\rho_1$ is such that $\E e^{\rho_1} = K_1>0$ and assume for simplicity that $\alpha_1>0$ is a constant. Then
\[
\bar f (x) = \E e^{\rho_1-\alpha_1 x} =  K_1 e^{-\alpha_1 x},
\]
$\bar F(x) = x \bar f (x)= xK_1 e^{-\alpha_1 x}$ and $\bar G(x) = x K_1 e^{-\alpha_1 x} - x$. By the analysis from Section \ref{s:stocdet} the deterministic maximum yield is achieved at the point $x_1$ where
\[
\bar G'(x_1)=\bar f(x_1) +x_1 \bar f'(x_1) -1 = K_1e^{-\alpha_1 x_1}-\alpha_1 x_1  K_1e^{-\alpha_1 x_1} - 1=0.
\]
Define the function
\[
q(x) = K_1 e^{-\alpha_1 x} - \alpha_1 K_1x e^{-\alpha_1 x} - 1, x\in \R_+.
\]
\begin{lm}
If $q(0)<0$ then the equation $q(x)=0$ has no solutions on $(0,\infty)$. If instead $q(0)>0$ then the equation $q(x)=0$ has exactly one solution $x_1>0$.
\end{lm}
\begin{proof}
Note that
\[
q'(x)= \alpha_1K_1 e^{-\alpha_1 x}(-2+\alpha_1x)
\]
and
\[
q''(x)=-\alpha_1^2 K_1  e^{-\alpha_1 x} (-2+\alpha_1 x) + \alpha_1^2 K_1 e^{-\alpha_1 x}.
\]
This shows that starting from $x=0$, the function $q$ decreases to its minimum at $x=\frac{2}{\alpha_1}$ and then increases from there on forever. However, once $q$ goes below zero it will never go above zero again. This happens because of the above properties and the fact that
\[
\lim_{x\to\infty} q(x) =\lim_{x\to\infty} (K_1 e^{-\alpha_1 x} - \alpha_1 K_1x e^{-\alpha_1 x} - 1)=-1.
\]
This implies that if $q(0)<0$ there are no solutions to $q(x)=0$. If we assume $q(0)>0$ we get in combination with $\lim_{x\to \infty }<0$ by the Intermediate Value Theorem that there exists a solution to $q(x)=0$. It is also clear by the properties of $q(x)$ that there exists exactly one solution to $q(x)=0$ and the solution has to lie in the interval $\left(0, \frac{2}{\alpha_1}\right)$.
\end{proof}
In order to be able to achieve this yield in the stochastic setting, according to Theorem \ref{t:maximal_yield} we need to ensure that $x_1$ is self-sustainable. This boils down to
\[
 \PP\left(f(x_1,\xi)\geq 1\right) =\PP\left(e^{\rho_1-\alpha_1 x_1}\geq 1\right)=1,
\]
or
\[
\PP\left(\rho_1\geq \alpha_1 x_1 \right)=1.
\]
Since $x_1\in \left(0, \frac{2}{\alpha_1}\right)$ we see that if $\rho_1\geq 2$ with probability one, then the self-sustaining harvesting policy given by
\begin{equation*}
   h^*(y) := \begin{cases}
             y-x_1 & \text{if } y>x_1,\\
               0 & \text{if } y\leq x_1,
          \end{cases}
\end{equation*}
where $x_1$ is the unique solution to $q(x)=0$, maximizes the expected long term yield and makes it equal to the deterministic maximal sustainable yield. The value of the optimal expected long term yield will be
\[
\E h^*(Y_\infty) = \bar G(x_1) = x_1(K_1 e^{-\alpha_1 x_1}-1).
\]
\subsection{Maximal constant effort policy.} Suppose we use a constant effort policy $h(x)=\theta x$ for some $\theta\in(0,1)$ and that both $\rho_1$ and $\alpha_1$ are random.
The condition for persistence (see Theorem \ref{t:pers1_disc} and Theorem \ref{t:pers_single3}) is given by
\[
\E\rho_1 + \ln(1-\theta)> 0.
\]
This forces that $\theta\in \left(0,1-e^{-\E\rho_1}\right)$. Assume this condition holds so that $Y_t$ converges to a stationary distribution $\pi_\theta$. Then \eqref{e:r0} becomes
\[
0 = \ln(1-\theta) +\E\rho_1 - (1-\theta) \E\alpha_1 \int_{(0,\infty)}x\,\pi_\theta(dx).
\]
We can use this to show that the long run expected yield is given by
\[
H(\theta):=\E h(Y_\infty) = \int_{(0,\infty)}h(x)\,\pi_\theta(dx) = \frac{\theta (\E\rho_1 +\ln(1-\theta))}{(1-\theta) \E\alpha_1}.
\]

 \begin{figure}[h!tb]
 	\begin{center}
 		\includegraphics[height=2in]{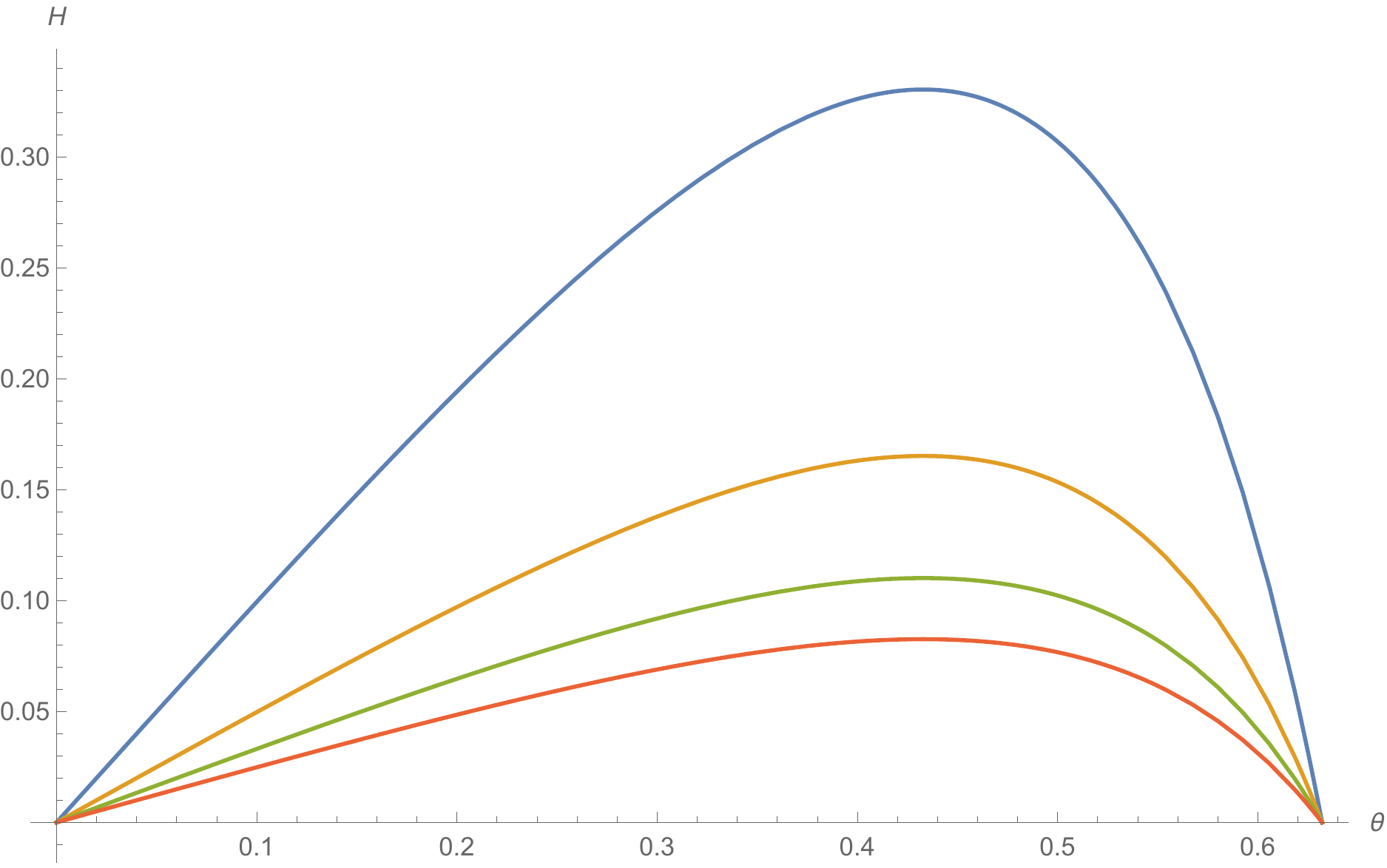}
 		\caption{Graph of the long run average yield $H(\cdot)$ as a function of the harvesting rate $\theta$ when $\E \rho_1=1$ and $\E\alpha_1=1, 2, 3, 4$.} \label{fig1}\end{center}
 \end{figure}

 The Intermediate Value Theorem shows there is a solution $\theta^*\in\left(0,1-e^{-\E\rho_1}\right)$ to
\begin{equation}\label{e:H'}
0=H'(\theta^*)=\frac{\E\rho_1-\theta^* + \ln(1-\theta^*)}{ \E\alpha_1(1-\theta^*)^2}.
\end{equation}
Since the function $p(x)=\ln(1-x)-x$ is strictly decreasing on $(0,1)$ we also get that the solution $x^*$ is unique.
Taking another derivative, evaluating at $x^*$ and using \eqref{e:H'} we get
\[
H''(\theta^*)= \frac{2\E\rho_1-\theta^*+ 2\ln(1-\theta^*)-2}{\E\alpha_1(1-\theta^*)^3} = \frac{\theta^*-2}{\E\alpha_1(1-\theta^*)^3}<0.
\]
This implies that $\theta^*$ is a global maximum of $H(\theta)$ on $\left[0,1-e^{-\E\rho_1}\right]$. The maximal expected constant effort harvesting yield will be
\[
H(\theta^*) =  \frac{(\theta^*)^2}{(1-\theta^*) \E\alpha_1}.
\]

\section{Harvesting of two interacting species}
In this section we analyze the situation when there are two interacting species that can be harvested.
The system is modelled in the absence of harvesting by
\begin{equation*}
\begin{aligned}
X^1_{t+1}&=X^1_{t} f_1(X^1_t, X^2_t,\xi_{t+1}),\\
X^2_{t+1}&=X^1_{t} f_2(X^1_t, X^2_t,\xi_{t+1}).
\end{aligned}
\end{equation*}
As the theory from Appendix \ref{s:a2d} shows, one needs to first look at the quantities
\[
r_i(\delta_0) = \E[\ln f_i(0,\xi_1)], i=1,2.
\]
If $r_i(\delta_0)>0$ then species $i$ survives on its own and converges to a unique invariant probability measure $\mu_i$ supported on $(0,\infty)$.
Suppose $r_1(\delta_0)>0, r_2(\delta_0)>0$. The realized per-capita growth rates can be computed via
  \[
  r_i(\mu_j)=\int_{(0,\infty)}\E[\ln f_i(x,\xi_1)]\mu_j(dx).
  \]
If $r_1(\mu_2)>0$ and $r_2(\mu_1)>0$ by Theorem \ref{t:pers1_disc} we have the convergence to a unique stationary distribution $\pi$ supported on $\Se_+$.
\subsection{Two species with harvesting.} Assume next that we harvest according to the strategies $h_1(x_1,x_2)$ and $h_2(x_1, x_2)$. Using \eqref{e:g} and \eqref{e:harv3} the dynamics becomes
\begin{equation}\label{e:harv2d}
Y^i_{t+1} = Y^i_t  g_i(\BY_t, \xi_{t+1}).
\end{equation}
where for $i=1,2$
\begin{equation}\label{e:g2}
   g_i(y_1,y_2,\xi) := \begin{cases}
              \frac{u_i(y_1,y_2)}{y_i} f_i(u_1(y_1,y_2), u_2(y_1,y_2), \xi)& \text{if } y_i>0,\\
               \left(\frac{\partial u_i}{\partial x_i}(y_1,y_2)\right)f_i(u_1(y_1,y_2), u_2(y_1,y_2), \xi) & \text{if } y_i = 0.
          \end{cases}
\end{equation}
Species $Y^i$ persists on its own with harvesting if
\begin{equation}\label{e:2d_pers2}
\begin{aligned}
r^Y_i(\delta_0) &= \E[\ln g_i(0,0,\xi_1)]\\
&= \E\ln\left[\left(\frac{\partial u_i}{\partial x_i}(0,0)\right)f_i(u_1(0,0), u_2(0,0), \xi_1)\right]\\
&= \ln \left(\frac{\partial u_i}{\partial x_i}(0,0)\right) + \E\ln f_i(0,0, \xi_1)\\
&>0.
\end{aligned}
\end{equation}
or equivalently
\begin{equation}\label{e:r_two}
\frac{\partial u_i}{\partial x_i}(0,0)> e^{-  \E\ln f_i(0,0, \xi_1)}.
\end{equation}
At this point there are three possibilities one may want to look at:
\begin{itemize}
  \item [1)] $\E\ln f_i(0,0, \xi_1)>0, i=1,2$ and $\frac{\partial u_i}{\partial x_i}(0,0)> e^{-  \E\ln f_i(0,0, \xi_1)}, i=1,2$ so that both harvested species persist on their own and have unique invariant probability measures $ \mu_1$ and $ \mu_2$ on the two positive axes. This describes the harvesting of a competitive system.
  \item [2)]$\E\ln f_i(0,0, \xi_1)>0, i=1,2$, $\frac{\partial u_1}{\partial x_1}(0,0)> e^{-  \E\ln f_1(0,0, \xi_1)}$ and $\frac{\partial u_2}{\partial x_2}(0,0)< e^{-  \E\ln f_2(0,0,\xi_1)}$. In this case, there are two species which compete with each other, both species persist on their own when there is no harvesting, and species $1$ persists with harvesting on its own, while species $2$ goes extinct if it is on its own and gets harvested.
  \item [3)] $\E\ln f_1(0,0, \xi_1)>0$, $\E\ln f_2(0,0, \xi_1)<0$, and $\frac{\partial u_1}{\partial x_1}(0,0)> e^{-  \E\ln f_1(0,0, \xi_1)}$. In this setting, species $1$ is a prey that persists on its own both with harvesting and without harvesting while species $2$ is a predator that can not persist on its own.
\end{itemize}
\begin{example}

Assume we work with constant threshold harvesting strategies, so that we harvest species $1$ according to
\begin{equation*}
   h(y_1,y_2) := \begin{cases}
             y_1-w & \text{if } y_1>w,\\
               0 & \text{if } y_1\leq w
          \end{cases}
\end{equation*}
where $w>0$. We will suppose species $2$ does not get harvested.
As we have seen in Section \ref{s:stocdet}, constant escapement harvesting strategies do not influence the persistence of a single species (as long as the threshold is strictly positive). However, we can show that they do change the persistence criteria if there are two interacting species.

Suppose species $1$ persists one its own: $\E\ln f_1(0,0, \xi_1)>0$. Without harvesting it converges to a stationary distribution $\tilde \mu_1$ while with harvesting it converges to a different stationary distribution $ \mu^{h}_1$.
Without harvesting we have
\[
r_2(\mu_1) = \int_{(0,\infty)}\E[\ln f_2(x,0,\xi_1)]\tilde \mu_1(dx)
\]
while with harvesting
\begin{equation}\label{e:2d_1}
\begin{aligned}
r_2(\mu_1^{h}) &= \int_{(0,\infty)}\E[\ln g_2(x,0,\xi_1)]\mu_1^{h}(dx) \\
&=  \int_{(0,\infty)}\E[\ln f_2(u_1(x_1,0),0,\xi_1)]\mu_1^{h}(dx_1)\\
&=\int_{(0,w)}\E[\ln f_2(x_1,0,\xi_1)]\mu_1^{h}(dx_1) + \int_{(w,\infty)}\E[\ln f_2(w,0,\xi_1)]\mu_1^{h}(dx_1).
\end{aligned}
\end{equation}
We see that in general, since $\mu_1^{h}\neq \tilde \mu_1$, we will have $r_2(\mu_1)\neq r_2(\mu_1^{h})$.  Therefore the persistence criteria are influenced by the threshold policies.
\end{example}

\subsection{Two dimensional  Lotka--Volterra predator-prey model.}
Suppose that we have model with a predator and a prey, that get harvested proportionally at rates $q,r\in (0,1)$. This means that $u_1(y_1,y_2)=(1-q)y_1$ and $u_2(y_1,y_2)=(1-r)y_2$. Using this in \eqref{e:harv2d} we get
\begin{eqnarray*}
Y^1_{t+1} &=& Y^1_t \exp \left(\rho_{t+1}+\ln(1-q)-\alpha_{t+1}(1-q)Y^1_t-a_{t+1}(1-r)Y^2_t\right)\\
Y^2_{t+1}&=& Y^2_t \exp \left(-d_{t+1}+\ln(1-r)-c_{t+1}(1-r)Y^2_t+b_{t+1}(1-q)Y^1_t\right)
\end{eqnarray*}
where the random coefficients have the following interpretations: $d_{t+1}>0$ is the predator's death rate, $a_{t+1}>0$ is the predator's attack rate on the prey, $b_{t+1}>0$ is the predator's conversion rate of prey, $c_{t+1}>0$ is the predator's intraspecific competition rate. Let $d_r:=\E[d_1]-\ln(1-r)>0, \rho_q:=\E[\rho_1]+\ln(1-q)$ and assume $(\rho_t)_{t\in \Z_+}, (\alpha_t)_{t\in \Z_+}, (a_t)_{t\in \Z_+}, (d_t)_{t\in \Z_+}, (c_t)_{t\in \Z_+}$ and $(b_t)_{t\in \Z_+}$ form independent sequences of i.i.d random variables. We assume for simplicity that the different random variables have compact support and are absolutely continuous with respect to Lebesgue measure. Then one can show by \cite{H87, BS19} that there is $K>0$ such that the process $\BY_t$ eventually enters and then stays forever in the compact set $\mathcal{K}=[0,K]^2$. We assume that the boundaries $\{0\}\times (0,K) $ and $(0,K)\times \{0\}$ are accessible for all $q,r\in(0,1)$. This can be easily checked for example if $\rho_1$ has $(0,L)$ for some $L>0$ in its support.

 As long as $\rho_q>0$, by our previous results, there exists a unique stationary distribution $\mu^q_1$ on $(0,\infty)\times \{0\}$ and
$$
 \int x_1\,\mu^q_1(dx_1) =\frac{\rho_q}{(1-q)\E \alpha_1} = \frac{\E[\rho_1]+\ln(1-q)}{(1-q)\E \alpha_1}.
$$
This can be used to get
$$
r_2(\mu^q_1) = -d_r +\E[b_1] \frac{\rho_q}{\E \alpha_1}= -(\E[d_1]-\ln(1-r)) +\E[b_1] \frac{\E \rho_1 + \ln(1-q)}{\E \alpha_1}.
$$
If $r_2(\mu_1^q)<0$ the predator $Y^2$ will go extinct with probability one. This means that, for a given harvesting rate $q\in \left(0, 1 - e^{-\E\rho_1}\right)$ of the prey, the maximal harvesting rate of the predator that does not lead to its extinction is
\[
r_{\max} = 1- \exp\left(\E d_1 - \E[b_1] \frac{\E \rho_1 + \ln(1-q)}{\E \alpha_1} \right).
\]

As long as $r_2(\mu_1^q)>0$, or equivalently
\[
\frac{d_r}{\rho_q}<  \frac{\E b_1}{\E \alpha_1},
\]
we get the existence of a unique invariant probability measure $\mu^{q,r}_{12}$ supported on a subset of $(0,\infty)^2$. Putting all the conditions together we get the following classification of the harvested dynamics:
\begin{itemize}
  \item If
   \begin{equation*}
\begin{split}
0&<q< 1 - e^{-\E\rho_1}\\
0&<r<r_{\max}=1- \exp\left(\E d_1 - \E[b_1] \frac{\E \rho_1 + \ln(1-q)}{\E \alpha_1} \right)
\end{split}
\end{equation*}
then the two species coexist and there is a unique invariant probability measure when
  \item If
   \begin{equation*}
\begin{split}
0&<q< 1 - e^{-\E\rho_1}\\
1&>r\geq r_{\max}=1- \exp\left(\E d_1 - \E[b_1] \frac{\E \rho_1 + \ln(1-q)}{\E \alpha_1} \right)
\end{split}
\end{equation*}
then the prey persists and the predator goes extinct with probability $1$.
  \item If $$1>q\geq  1 - e^{-\E\rho_1}$$
then both the prey and the predator go extinct with probability $1$.
\end{itemize}
We depict one example of the three possible regions in Figure \ref{fig2}.
\begin{figure}[h!tb]
 	\begin{center}
 		\includegraphics[height=2.8in]{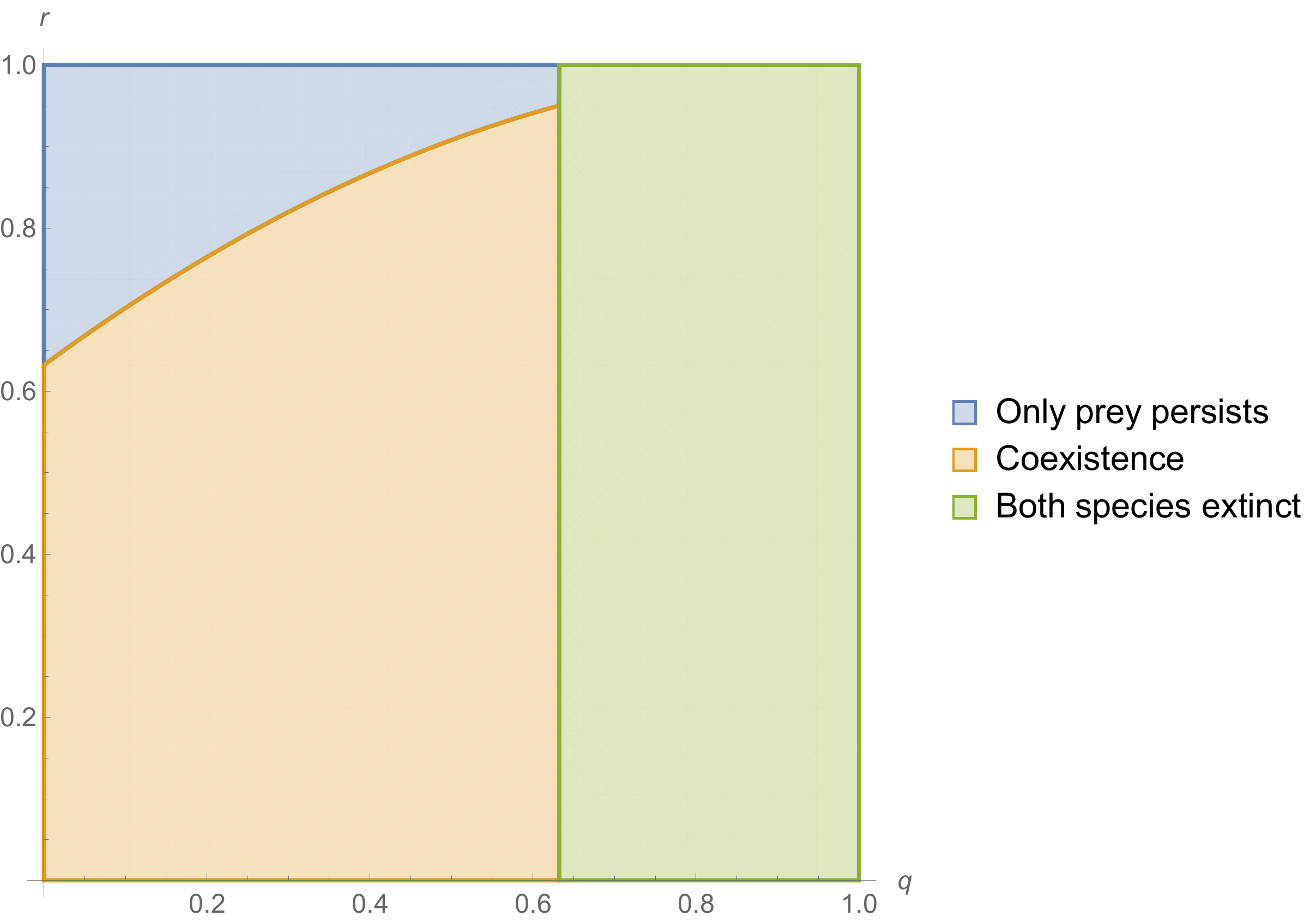}
 		\caption{The regions of the harvesting rates $q,r$ for which both species persist, for which just the prey persists and for which both species go extinct. The parameters are $\E a_1=\E d_1= \E \alpha_1 = \E \rho_1 =1, \E b_1=2, \E c_1=1.5$.} \label{fig2}\end{center}
 \end{figure}
Suppose next that the two species persist. If we set $\bar Y^1 := \int_{(0,\infty)^2} x_1\mu^{q,r}_{12}(du), \bar Y^2 := \int_{(0,\infty)^2} x_2\mu^{q,r}_{12}(du)$, since by \eqref{e:equ} the per-capita growth rates at stationarity are zero, we get
\begin{eqnarray*}
r_1(\mu^{q,r}_{12})&=& 0 = \rho_q-\E[\alpha_{1}](1-q)\bar Y^1-\E[a_{1}](1-r)\bar Y^2\\
r_2(\mu^{q,r}_{12}) &=&0=-d_r-\E[c_{1}](1-r)\bar Y^2+(1-q)\E[b_{1}]\bar Y^1.
\end{eqnarray*}
Solving this linear system we get the unique solution
\begin{equation}\label{e:Y}
\begin{split}
\bar Y^1&= \frac{\E[c_1]\rho_q+\E[a_1]d_r}{1-q}\frac{1}{\E[\alpha_1]\E[c_1]+\E[a_1]\E[b_1]}\\
\bar Y^2&=\frac{\E[b_1]\rho_q-\E[\alpha_1]d_r}{1-r}\frac{1}{\E[\alpha_1]\E[c_1]+\E[a_1]\E[b_1]}
\end{split}
\end{equation}

On the other hand, if $r_2(\mu_1^q)<0$ we get by the results from Appendix \ref{s:a2d} that $\lim_{t\to\infty}\frac{\ln Y_t^2}{t}=r_2(\mu_1^q)<0$ so that the predator goes extinct exponentially fast.
\begin{prop}\label{p:1}
If $r_2(\mu_1^q)<0$ and $r_1(\delta_0)=\rho_q>0$ the prey will persist and $\E Y_t^1\to \frac{\rho_q}{(1-q)\E[\alpha_1]}$ as $t\to\infty$.
\end{prop}
\begin{proof}Pick $\varepsilon>0$ small. Since $Y^2_t<\varepsilon$ for $t>T$ large enough, we have with high probability
\begin{equation}\label{e:conv}
\begin{split}
   Y^1_t e^{ \rho_{t+1}+\ln(1-q)-\alpha_{t+1}(1-q)Y^1_t-a_{t+1}(1-r)\varepsilon} &<Y^1_{t+1} \\
   &= Y^1_t e^{ \rho_{t+1}+\ln(1-q)-\alpha_{t+1}(1-q)Y^1_t-a_{t+1}(1-r)Y^2_t}  \\
     &<Y^1_t e^{\rho_{t+1}+\ln(1-q)-\alpha_{t+1}(1-q)Y^1_t}.
     \end{split}
\end{equation}
Define the processes $(Y_t^\varepsilon)_{t\in\Z_+}$ and $(\tilde Y_t)_{t\in\Z_+}$ via
\[
Y_{t+1}^\varepsilon= Y^\varepsilon_t e^{ \rho_{t+1}+\ln(1-q)-\alpha_{t+1}(1-q)Y^\varepsilon_t-a_{t+1}(1-r)\varepsilon}
\]
and
\[
\tilde Y^1_{t+1}=\tilde Y^1_t e^{\rho_{t+1}+\ln(1-q)-\alpha_{t+1}(1-q)\tilde Y^1_t}.
\]
By the results from Appendix \ref{s:a1} it is easy to see that
\[
\lim_{t\to\infty} \E \tilde Y^1_{t+1}= \frac{\E \rho_1 + \ln(1-q)}{\E \alpha_1}
\]
and
\[
\lim_{t\to\infty} \E Y^\varepsilon_{t+1}= \frac{\E \rho_1 + \ln(1-q)-\varepsilon}{\E \alpha_1}.
\]
By \eqref{e:conv} we get
\[
\frac{\E \rho_1 + \ln(1-q)-\varepsilon}{\E \alpha_1}=\lim_{t\to\infty} \E \tilde Y^\varepsilon_{t+1}\leq \limsup_{t\to\infty} \E Y_t^1 \leq \lim_{t\to\infty} \E \tilde Y^1_{t+1}= \frac{\E \rho_1 + \ln(1-q)}{\E \alpha_1}
\]
and
\[
\frac{\E \rho_1 + \ln(1-q)-\varepsilon}{\E \alpha_1}=\lim_{t\to\infty} \E \tilde Y^\varepsilon_{t+1}\leq \liminf_{t\to\infty} \E Y_t^1 \leq \lim_{t\to\infty} \E \tilde Y^1_{t+1}= \frac{\E \rho_1 + \ln(1-q)}{\E \alpha_1}.
\]
Letting $\varepsilon\downarrow 0$ yields
\[
\lim_{t\to\infty}\E Y_t^1= \frac{\rho_q}{(1-q)\E[\alpha_1]}
\]
which finishes the proof.
\end{proof}

Suppose we want to find the optimal harvesting strategy. Since the profit from harvesting prey or predators might be different we let $\beta> 0$ represent the relative value of the predator compared to the prey.
The problem then becomes maximizing the function
$$
H(q,r):=\lim_{T\to \infty} \E[h_1( Y^1_T)+\beta h_2(Y^2_T)] = \lim_{T\to \infty} \E[q Y^1_T+\beta r Y^2_T]=\lim_{T\to \infty} \frac{\sum_{n=1}^T q Y^1_n+r\beta Y^2_n}{T} = q \bar Y^1 + r\beta \bar Y^2,
$$
for $q,r\in[0,1]^2$.  Using the expressions for $\bar Y^1$, $\bar Y^2$ from \eqref{e:Y} together with Proposition \ref{p:1}, with the understanding that we set $\bar Y^1=\frac{\rho_q}{(1-q)\E[\alpha_1]}, \bar Y^2=0$ if the prey persists and the predator goes extinct, and the domain regions identified above we get
$$
H(q,r) = \left\{
        \begin{array}{ll}
            \left(\frac{\E[c_1]q\rho_q+\E[a_1]qd_r}{1-q} + \frac{\beta\E[b_1]r\rho_q-\beta\E[\alpha_1]rd_r}{1-r}\right)\frac{1}{\E[\alpha_1]\E[c_1]+\E[a_1]\E[b_1]} & \quad  r<r_{\max}, 0<q< 1 - e^{-\E\rho_1} \\
           \frac{q\rho_q}{(1-q)\E[\alpha_1]} & \quad r\geq r_{\max}, 0<q< 1 - e^{-\E\rho_1}\\
           0&\quad q\geq 1 - e^{-\E\rho_1}.
        \end{array}
    \right.
$$

\begin{figure}[h!tb]
 	\begin{center}
 		\includegraphics[height=2.8in]{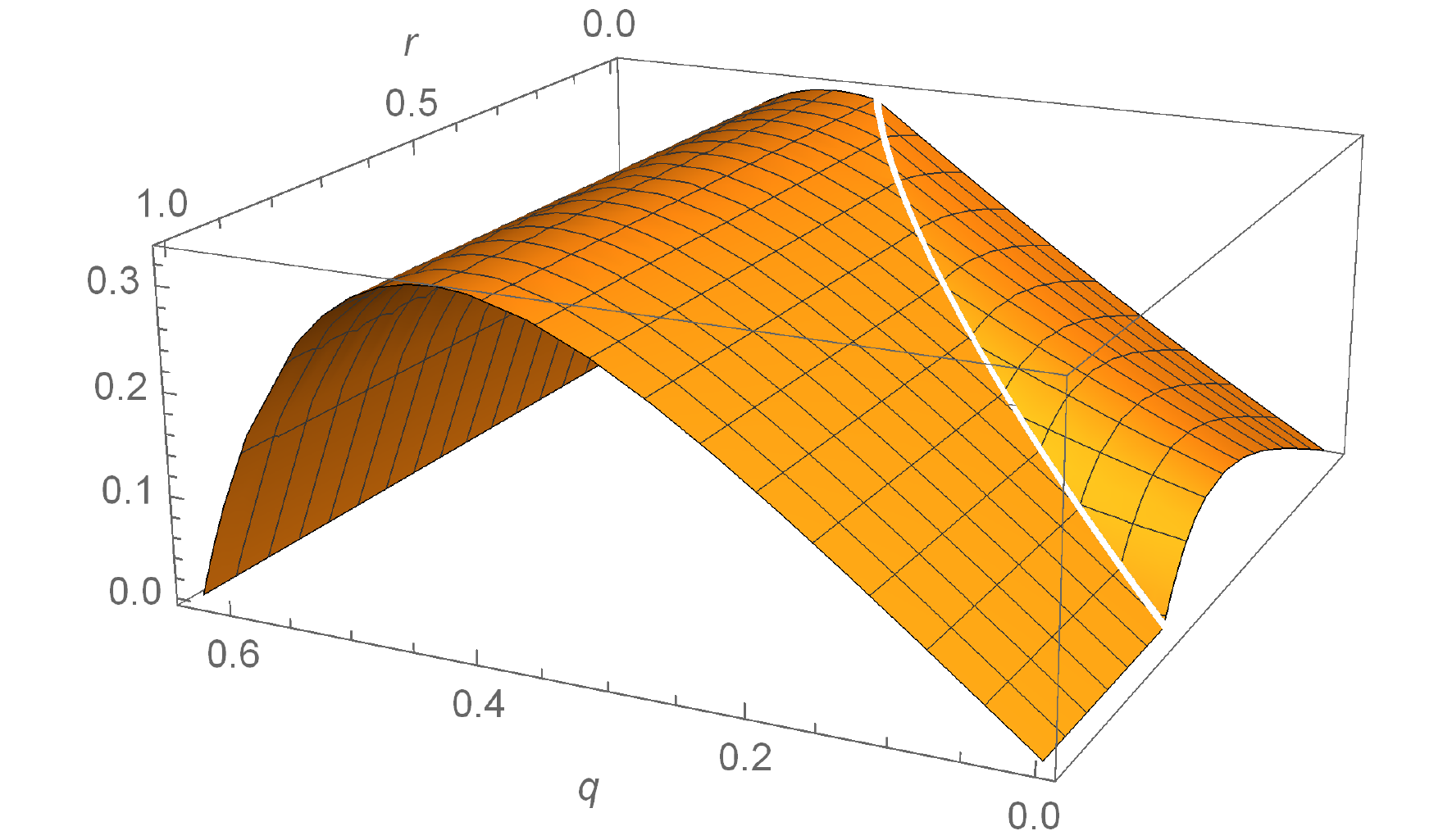}
 		\caption{The graph of $H(q,r)$. The parameters are $\beta = \E a_1=\E d_1= \E \alpha_1 = \E \rho_1 =1, \E b_1=2, \E c_1=1.5$.} \label{fig3}\end{center}
 \end{figure}
\begin{figure}[h!tb]
 	\begin{center}
 		\includegraphics[height=2.8in]{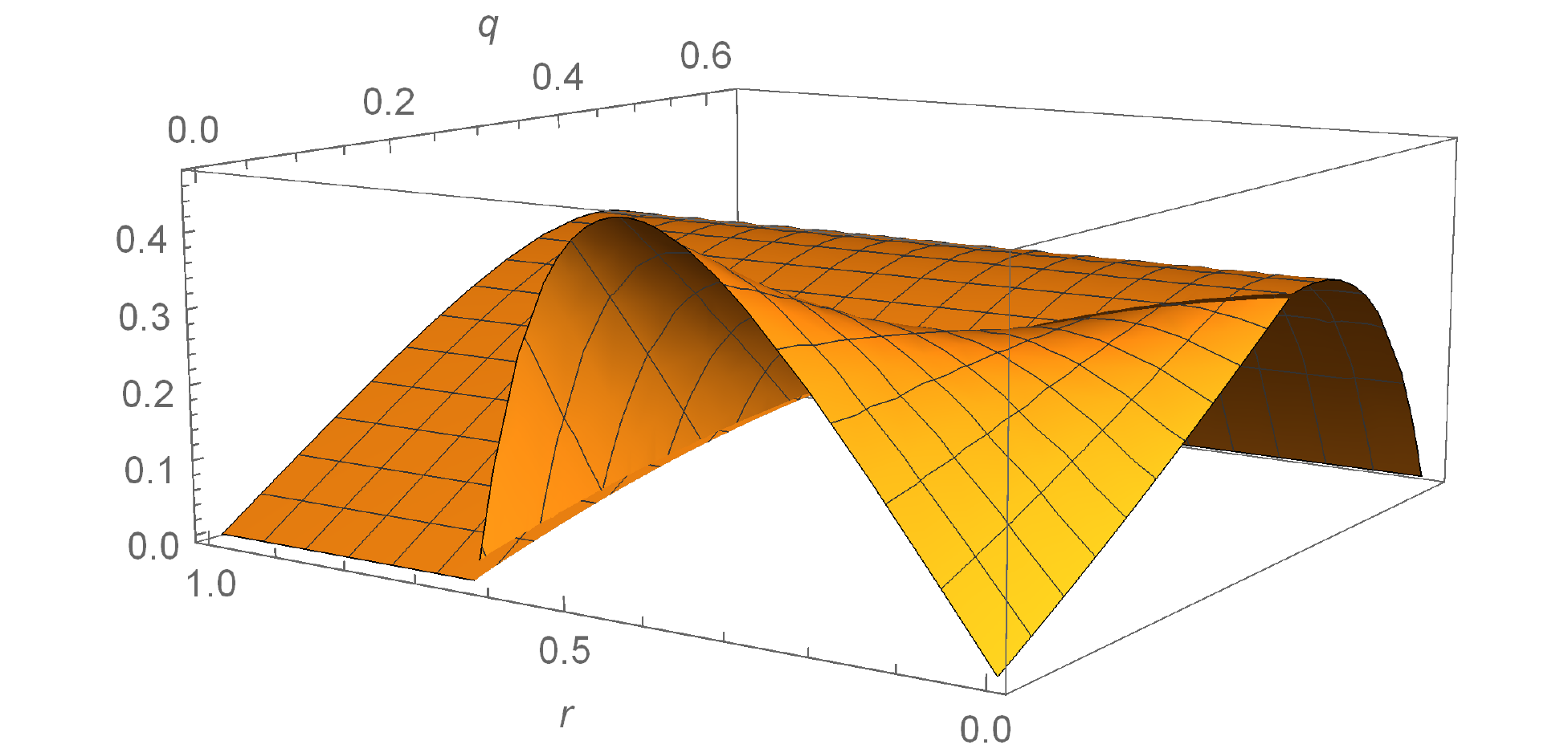}
 		\caption{The graph of $H(q,r)$. The parameters are $\beta = 5, \E a_1=\E d_1= \E \alpha_1 = \E \rho_1 =1, \E b_1=2, \E c_1=1.5$.} \label{fig4}\end{center}
 \end{figure}

 \textit{\textbf{Biological interpretation:} It is never optimal to harvest both the predator and the prey. If the relative price $\beta$ of the predator compared to the prey is low, it is always optimal to harvest the predator to extinction (see Figure \ref{fig3}). This then lets the prey population increase, and one gains by harvesting the prey. If instead the relative price $\beta$ is high, it is optimal to never harvest the prey (see Figure \ref{fig4}). This leads to an increase in the predator population, which then increases the harvesting yield of the predators. }

\subsection{Two dimensional Lotka-Volterra competition mode.}
We look at a two-species discrete Lotka-Volterra competition model when the two species get harvested proportionally at rates $q,r\in (0,1)$. The harvested dynamics is given by
\begin{eqnarray*}
Y^1_{t+1} &=& Y^1_t \exp \left(\rho^1_{t+1}+\ln(1-q)-\alpha_{t+1}(1-q)Y^1_t-a_{t+1}(1-r)Y^2_t\right)\\
Y^2_{t+1}&=& Y^2_t \exp \left(\rho^2_{t+1}+\ln(1-r)-c_{t+1}(1-r)Y^2_t-b_{t+1}(1-q)Y^1_t\right)
\end{eqnarray*}
We set $\bar\rho^1_q:=\E[\rho^1_1]+\ln(1-q)$ and $\bar\rho^2_q:=\E[\rho^2_1]+\ln(1-r)>0$. We assume the different random coefficients are independent and form sequences of i.i.d random variables. Furthermore, we make the same assumptions that were made in the predator-prey system. These ensure that the state space is compact and that the boundaries are accessible.

 As long as $\bar\rho^1_q, \bar\rho^2_r>0$, by our previous results, there exists a unique stationary distribution $\mu^q_1$ (respectively $\mu^r_2$) on $(0,\infty)\times \{0\}$ (respectively $\{0\}\times (0,\infty) $) and
$$
 \int x_1\,\mu^q_1(dx_1) =\frac{\bar\rho^1_q}{(1-q)\E \alpha_1} = \frac{\E[\rho^1_1]+\ln(1-q)}{(1-q)\E \alpha_1},
$$
$$
 \int x_2\,\mu^r_2(dx_2) =\frac{\bar\rho^2_r}{(1-r)\E c_1} = \frac{\E[\rho^2_1]+\ln(1-r)}{(1-r)\E c_1}.
$$
One can then compute the per-capita growth rates
$$
r_1(\mu^r_2) = \bar \rho_q^1 - \E a_1 (1-r) \int x_2\,\mu^r_2(dx_2) = (\E[\rho^1_1]+\ln(1-q)) -\E a_1 \frac{\E \rho^2_1 + \ln(1-r)}{\E c_1},
$$
and
$$
r_2(\mu^q_1) = \bar \rho_r^2 - \E b_1 (1-q) \int x_1\,\mu^q_1(dx_1) = (\E[\rho^2_1]+\ln(1-r)) -\E b_1 \frac{\E \rho^1_1 + \ln(1-q)}{\E \alpha_1}.
$$

We get the following classification of the dynamics:
\begin{itemize}
  \item If
   \begin{equation*}
\begin{split}
0&<q< 1 - e^{-\E\rho^1_1}\\
0&<r< 1 - e^{-\E\rho^2_1}\\
 \frac{\E a_1}{\E c_1} &<  \frac{\E \rho^1_1 +\ln(1-q)}{\E \rho^2_1 + \ln(1-r)} < \frac{\E \alpha_1}{\E b_1}
\end{split}
\end{equation*}
then the two species coexist and the process converges to its unique invariant probability measure.
  \item If
  \begin{equation*}
\begin{split}
0&<q< 1 - e^{-\E\rho^1_1}\\
0&<r< 1 - e^{-\E\rho^2_1}\\
 \frac{\E a_1}{\E c_1} &<  \frac{\E \rho^1_1 +\ln(1-q)}{\E \rho^2_1 + \ln(1-r)}\\
  \frac{\E b_1}{\E \alpha_1} &>  \frac{\E \rho^2_1 +\ln(1-r)}{\E \rho^1_1 + \ln(1-q)}
\end{split}
\end{equation*}
then species 1 persists and species 2 goes extinct with probability $1$.
 \item If
   \begin{equation*}
\begin{split}
0&<q< 1 - e^{-\E\rho^1_1}\\
0&<r< 1 - e^{-\E\rho^2_1}\\
 \frac{\E a_1}{\E c_1} &>  \frac{\E \rho^1_1 +\ln(1-q)}{\E \rho^2_1 + \ln(1-r)}\\
  \frac{\E b_1}{\E \alpha_1} &<  \frac{\E \rho^2_1 +\ln(1-r)}{\E \rho^1_1 + \ln(1-q)}
\end{split}
\end{equation*}
then species 2 persists and species 1 goes extinct with probability $1$.
 \item If $(Y_0^1, Y_0^2)=(x,y)\in (0,\infty)^2$ then let $p_{x,y} = \PP(Y_t^1\to \mu_1^q, Y_t^2\to 0~|~(Y_0^1, Y_0^2)=(x,y))$. If
\begin{equation*}
\begin{split}
0&<q< 1 - e^{-\E\rho^1_1}\\
0&<r< 1 - e^{-\E\rho^2_1}\\
 \frac{\E a_1}{\E c_1} &>  \frac{\E \rho^1_1 +\ln(1-q)}{\E \rho^2_1 + \ln(1-r)}\\
  \frac{\E b_1}{\E \alpha_1} &>  \frac{\E \rho^2_1 +\ln(1-r)}{\E \rho^1_1 + \ln(1-q)}
\end{split}
\end{equation*}
then we have bistability, that is, $p_{x,y}\in (0,1)$ and $1-p_{x,y}= \PP(Y_t^2\to \mu_2^r, Y_t^1\to 0~|~(Y_0^1, Y_0^2)=(x,y))$.
\end{itemize}
Note that for a given set of coefficients one cannot have all the 4 regions if we vary $q$ and $r$. There are two possibilities, each having three regions. One is to have coexistence, the persistence of species 1 and extinction of species 2, or the extinction of species 1 and the persistence of species 2 (see Figure \ref{fig5}). The other possibility is to have bistability, the persistence of species 1 and extinction of species 2, or the extinction of species 1 and the persistence of species 2 (see Figure \ref{fig6}).

\textit{\textbf{Biological interpretation:} Harvesting can facilitate the coexistence of the two species. For example, suppose that
\[
 \frac{\E a_1}{\E c_1} < \frac{\E \alpha_1}{\E b_1}<  \frac{\E \rho^1_1 }{\E \rho^2_1 }.
\]
This implies that if there is no harvesting species $1$ persists and species $2$ goes extinct. It is clear that there exists $q\in (0,1)$ such that
\[
 \frac{\E a_1}{\E c_1}<  \frac{\E \rho^1_1 +\ln(1-q)}{\E \rho^2_1  } < \frac{\E \alpha_1}{\E b_1}
\]
which leads to coexistence. If one species has a competitive advantage so that without harvesting it drives the other competitor extinct, one can harvest this dominant species and get coexistence.
}
\begin{figure}[h!tb]
 	\begin{center}
 		\includegraphics[height=2.8in]{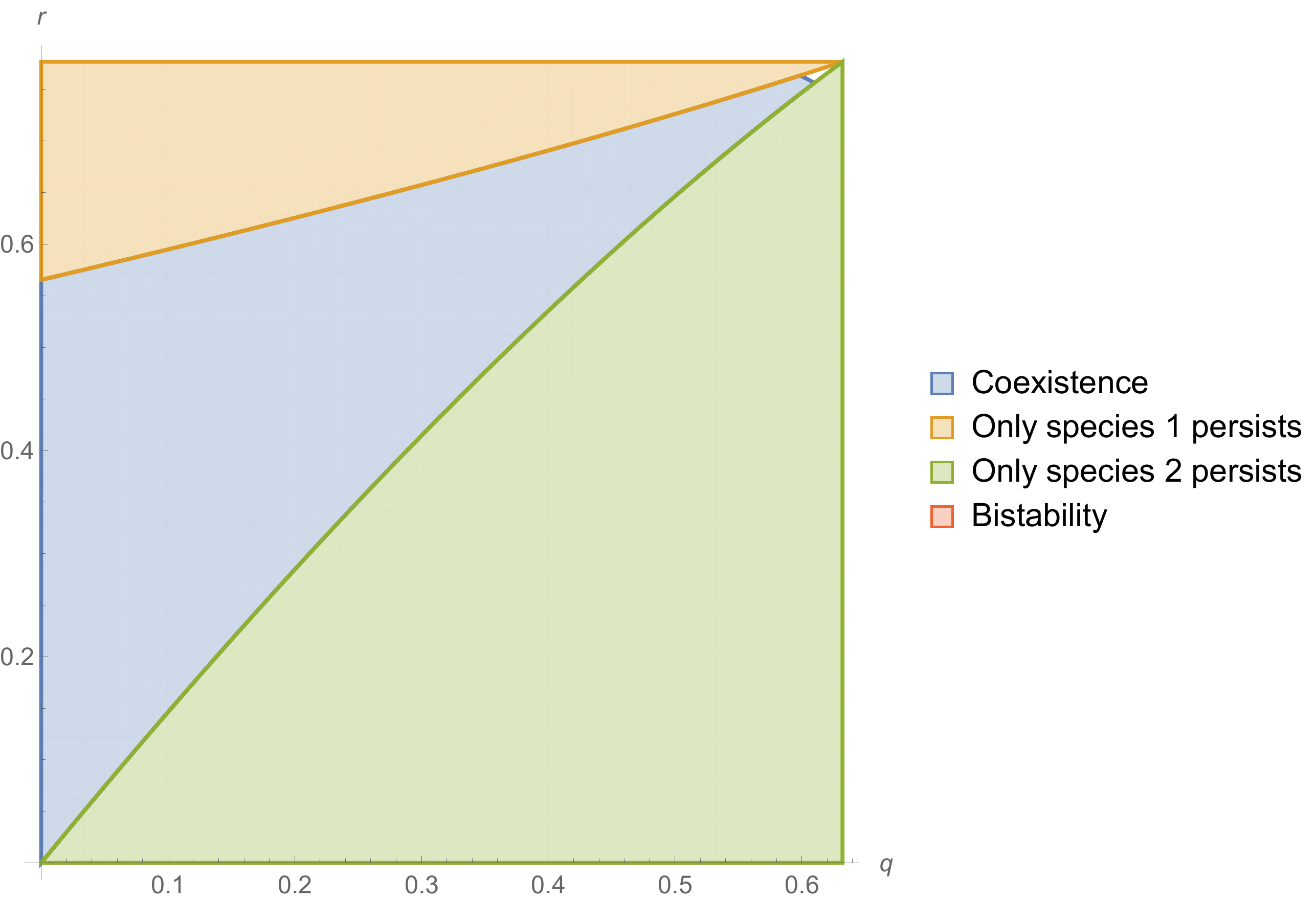}
 		\caption{Figure showing the regions of the harvesting rates $q,r$ for which both species persist, for which just the species 1 persists, and the region for which just species 2 persists. The parameters are $\E a_1=\E b_1=1, \E c_1=\E \alpha_1=1.5, \E \rho^1_1 =1, \E \rho^2_1 =1.5$.} \label{fig5}\end{center}
 \end{figure}

\begin{figure}[h!tb]
 	\begin{center}
 		\includegraphics[height=2.8in]{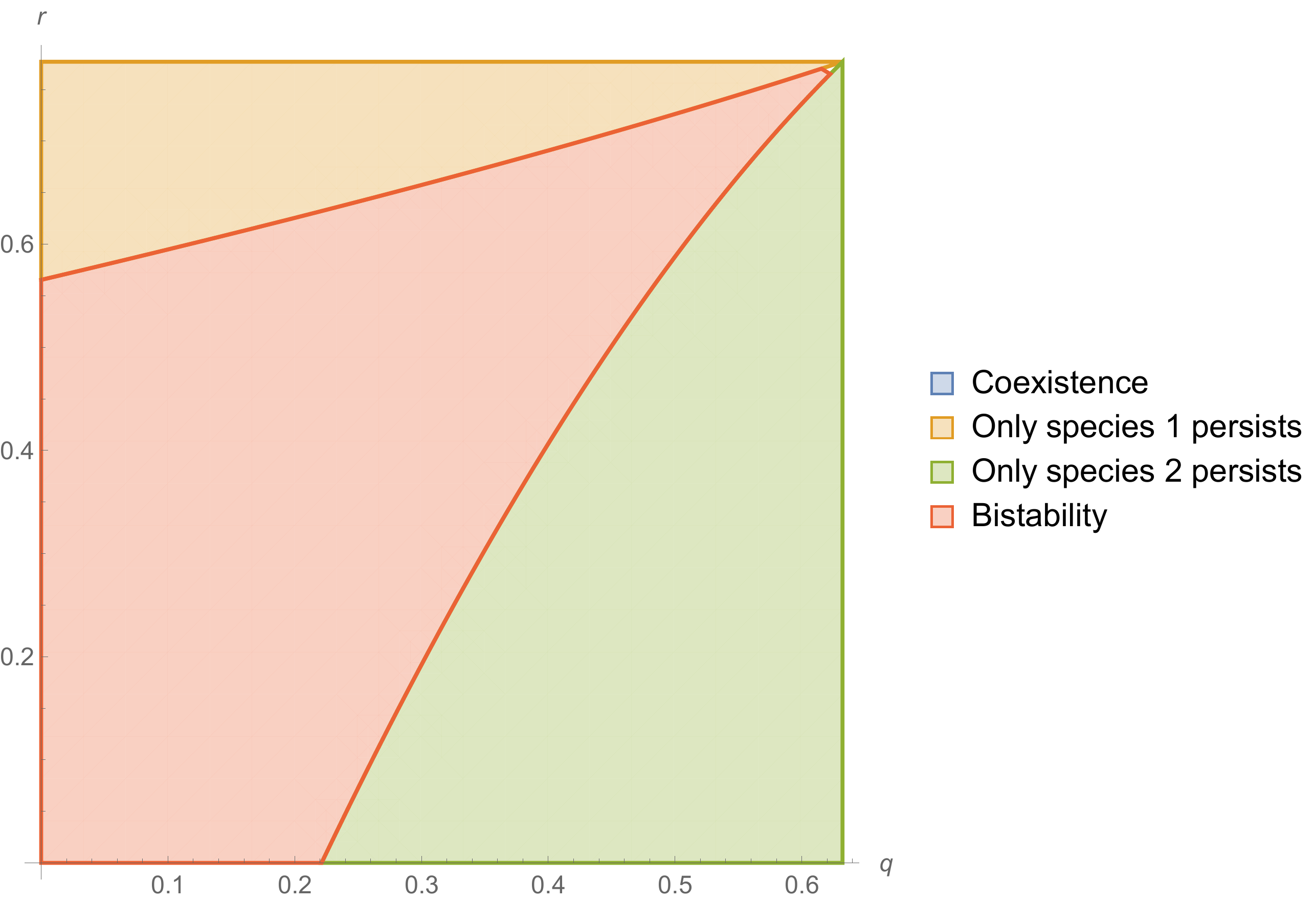}
 		\caption{Figure showing the regions of the harvesting rates $q,r$ for which there is bistability, for which just the species 1 persists, and the region for which just species 2 persists. The parameters are $\E a_1=1.5, \E b_1=2, \E c_1=1, \E \alpha_1=1, \E \rho^1_1 =1, \E \rho^2_1 =1.5$.} \label{fig6}\end{center}
 \end{figure}
If there is coexistence and the system converges to an invariant probability measure $\mu_{12}^{q, r}$ on $(0,\infty)^2$ we see by \eqref{e:equ} that the per-capita growth rates at stationarity are zero. This shows that
\begin{eqnarray*}
r_1(\mu_{12})=0&=& \bar \rho_q^1-\E[\alpha_{1}](1-q)\bar Y^1-\E[a_{1}](1-r)\bar Y^2\\
r_2(\mu_{12})=0&=&  \bar \rho_r^2-\E[c_{1}](1-r)\bar Y^2-(1-q)\E[b_{1}]\bar Y^1,
\end{eqnarray*}
where as before $\bar Y^1 = \int_{(0,\infty)^2} x_1\mu^{q,r}_{12}(du), \bar Y^2= \int_{(0,\infty)^2} x_2\mu^{q,r}_{12}(du)$ are the expected values of the two species at stationarity.
Solving this linear system yields the unique solution
\begin{eqnarray*}
\bar Y^1&=& \frac{\E[c_1]\bar \rho^1_q-\E[a_1]\bar\rho^2_r}{1-q}\frac{1}{\E[\alpha_1]\E[c_1]-\E[a_1]\E[b_1]}\\
\bar Y^2&=&\frac{\E[\alpha_1]\bar \rho^2_r-\E[b_1]\bar \rho^1_q}{1-r}\frac{1}{\E[\alpha_1]\E[c_1]-\E[a_1]\E[b_1]}.
\end{eqnarray*}
One can prove the following analogue of Proposition \ref{p:1}.
\begin{prop}\label{p:2}
If $r_2(\mu^q_1)<0$ and $r_1(\mu^r_2)>0$ then species 1 will persist and $\E Y_t^1\to \frac{\bar\rho^1_q}{(1-q)\E[\alpha_1]}$ as $t\to\infty$.
If $r_1(\mu^r_2) <0$ and $r_2(\mu^q_1)>0$ then species 2 will persist and $\E Y_t^2\to \frac{\bar \rho^2_r}{(1-r)\E[c_1]}$ as $t\to\infty$.
\end{prop}
Suppose the coefficients are such that the coexistence of the two species is possible. We are interested in maximizing the function
$$
H(q,r)=\lim_{T\to \infty} \E[q Y^1_T+\beta rY^2_T] =\lim_{T\to \infty} \frac{\sum_{n=1}^T q Y^1_n+r\beta Y^2_n}{T} = q \bar Y^1 + r\beta \bar Y^2,
$$
where $\beta> 0$ represents the relative value of species 2 compared to species 1. Using the expressions for $\bar Y^1$, $\bar Y^2$, together with Proposition \ref{p:2} and the domain regions identified above we get
$$
H(q,r) = \left\{
        \begin{array}{ll}
            \left(q\frac{\E[c_1]\bar \rho^1_q-\E[a_1]\bar\rho^2_r}{1-q} + \beta r\frac{\E[\alpha_1]\bar \rho^2_r-\E[b_1]\bar \rho^1_q}{1-r}\right)\frac{1}{\E[\alpha_1]\E[c_1]-\E[a_1]\E[b_1]} & \quad  r_2(\mu^q_1)>0, r_1(\mu^r_2)>0 \\
          q \frac{\bar\rho^1_q}{(1-q)\E[\alpha_1]} & \quad r_2(\mu^q_1)<0, r_1(\mu^r_2)>0\\
           \beta r\frac{\bar\rho^2_r}{(1-r)\E[c_1]} & \quad r_2(\mu^q_1)>0, r_1(\mu^r_2)<0\\
           0&\quad q\geq 1 - e^{-\E\rho^1_1}, r\geq 1 - e^{-\E\rho^2_1}.
        \end{array}
    \right.
$$

\begin{figure}[h!tb]
 	\begin{center}
 		\includegraphics[height=2.8in]{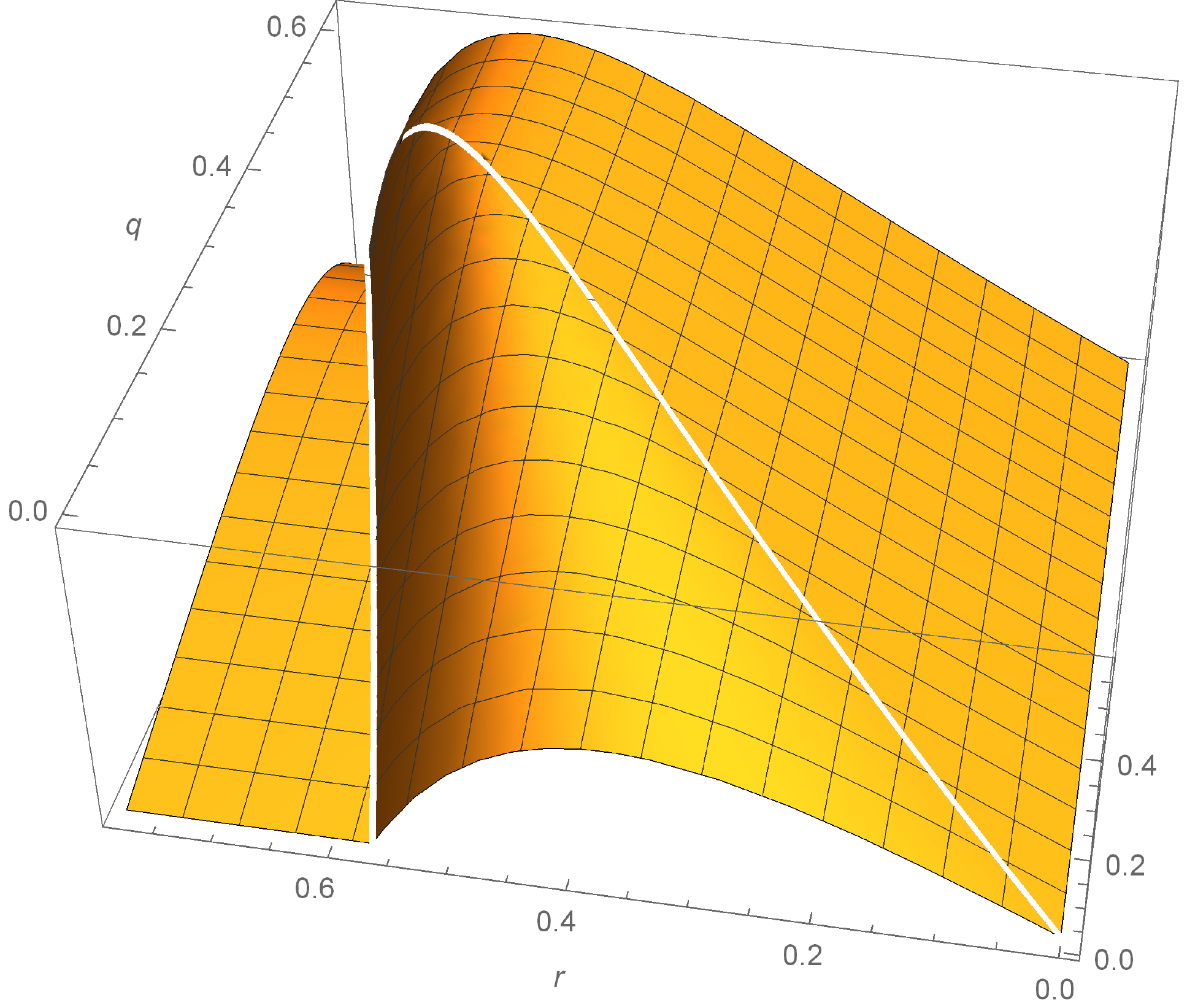}
 		\caption{The graph of $H(q,r)$. The parameters are $\beta=1, \E a_1=\E b_1=1, \E c_1=\E \alpha_1=1.5, \E \rho^1_1 =1, \E \rho^2_1 =1.5$} \label{fig7}\end{center}
 \end{figure}

\begin{figure}[h!tb]
 	\begin{center}
 		\includegraphics[height=2.8in]{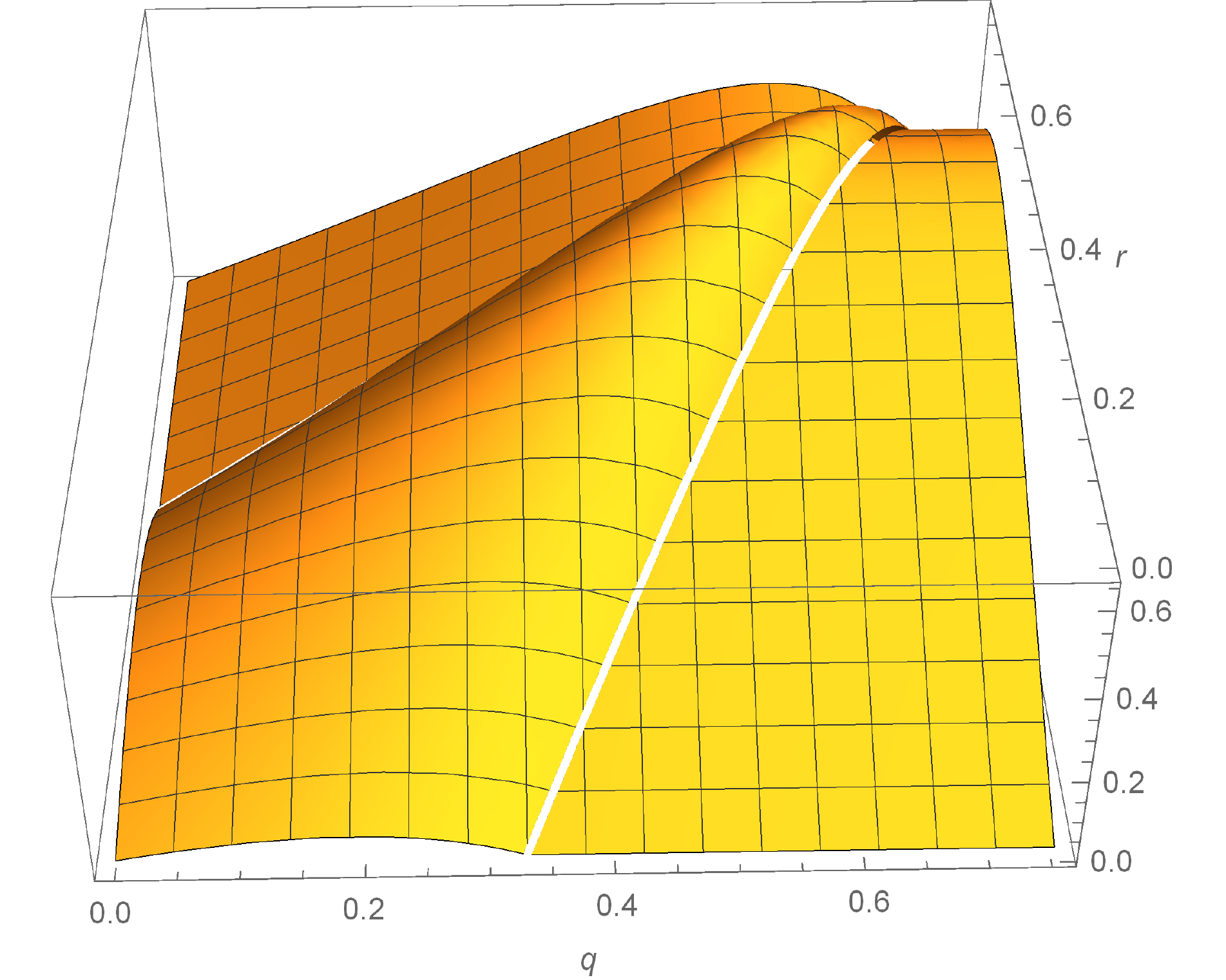}
 		\caption{The graph of $H(q,r)$. The parameters are $\beta=1, \E a_1=\E b_1=1, \E c_1=\E \alpha_1=1.5, \E \rho^1_1 =1.4, \E \rho^2_1 =1.5$} \label{fig8}\end{center}
 \end{figure}

 \begin{figure}[h!tb]
 	\begin{center}
 		\includegraphics[height=2.8in]{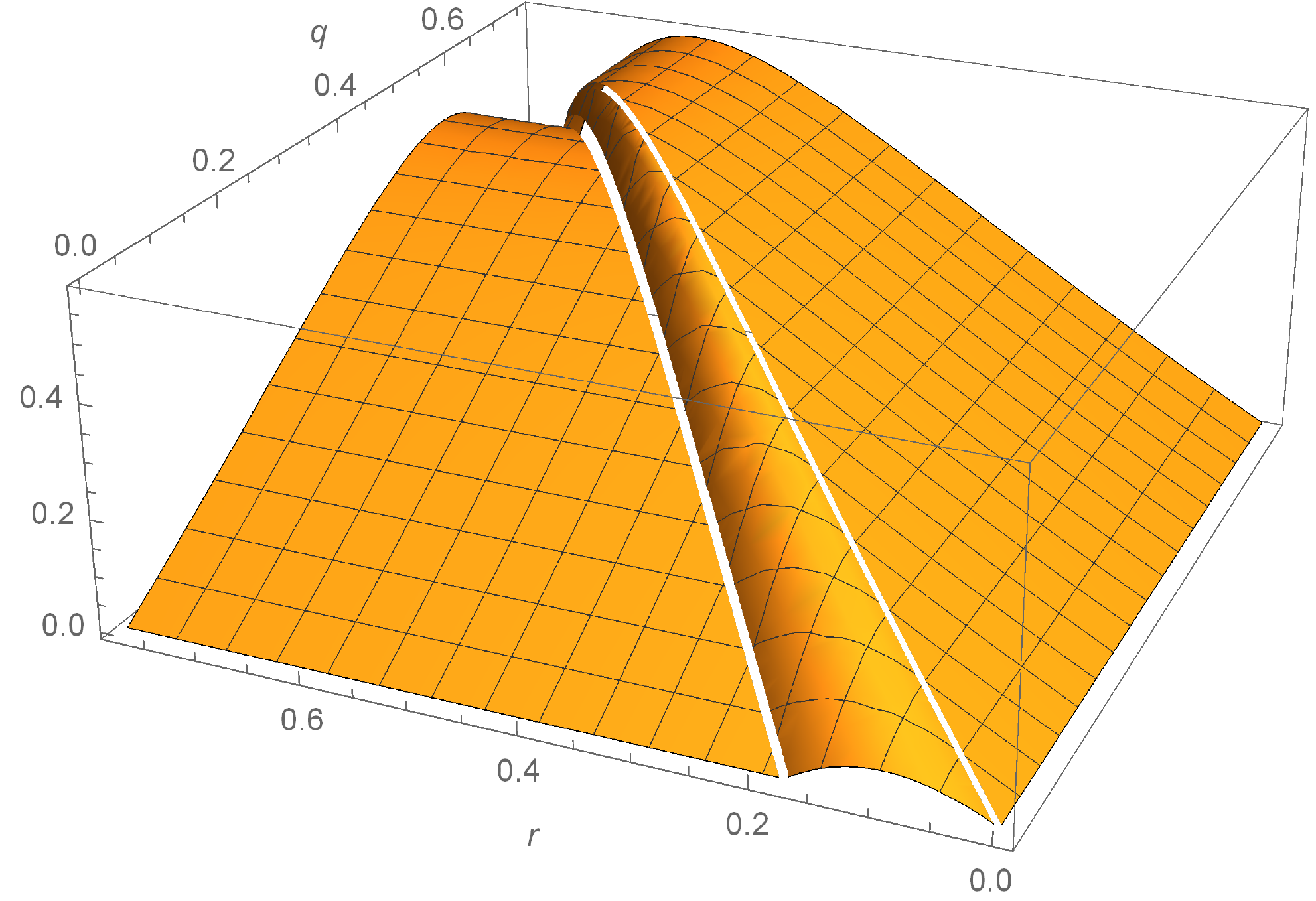}
 		\caption{The graph of $H(q,r)$. The parameters are $\beta=1, \E a_1=\E b_1=1.4, \E c_1=\E \alpha_1=1.5, \E \rho^1_1 =1.4, \E \rho^2_1 =1.5$} \label{fig9}\end{center}
 \end{figure}
\textit{\textbf{Biological interpretation:} Depending on the interaction coefficients, growth rates, and the relative value of the species there are three possible scenarios for the optimal harvesting strategy. In one case we harvest species $1$ to extinction and maximize the yield from harvesting species $2$. In other instances it is best to harvest species $2$ to extinction and maximize the harvest from species $1$. The third instance is the one of coexistence: the optimal harvesting strategy is to keep both species alive. In Figure \ref{fig7} we can see, that since the growth rate of species $2$ is greater than that of species $1$, while the other coefficients are identical, it is optimal to harvest species $1$ to extinction and to get a higher harvesting yield from species $2$. In the example from Figure \ref{fig8}, when the species are similar to each other, it is optimal to keep both species alive. However, once we increase the competition, it becomes optimal to drive one species extinct through harvesting (see Figure \ref{fig9}). These examples show that there is a delicate balance one has to take into account when looking for the optimal harvesting strategies. The intra- and interspecific competition rates, growth rates, and the prices of the species turn out to play key roles.
}

\textbf{Acknowledgements:} The author thanks Dang Nguyen and Sergiu Ungureanu for helpful discussions related to the paper.
\bibliographystyle{agsm}
\bibliography{harvest}

\appendix
\section{Criteria for persistence and extinction}\label{s:a1}
\subsection{Single species system}
Suppose we have one species whose dynamics is given by
\begin{equation}\label{e:single11}
X_{t+1} = X_t  f(X_t, \xi_{t+1})
\end{equation}
We present a few known results which give the existence of a unique invariant probability measure. These results appear in work by \cite{E84, E89,VH97, FH02, S12}.
\begin{thm}\label{t:pers_single1}
Assume that $F(x,\xi)=xf(x,\xi)$ is continuously differentiable and strictly increasing in $x$, and $f(x,\xi)$ is strictly decreasing in $x$. If $\E[\ln f(0,\xi_1)]>0$ and $\lim_{x\to\infty} \E[\ln f(x,\xi_1)]<0$, then there exists a positive invariant probability measure $\mu$ and the distribution of $X_t$ converges weakly to $\mu$ whenever $X_0=x>0$.
\end{thm}
Sometimes, if monotonicity fails, one can make use of the following result (\cite{VH97}).
\begin{thm}\label{t:pers_single2}
Assume that $$f(x,\xi)=\lambda h(x)^{-\xi}$$ where $g$ is a positive differentiable function such that $x\mapsto xh'(x)/h(x)$ is strictly increasing on $[0,\infty)$. Assume $\E \xi_1, \E \xi_1^2 <\infty$ and $\xi_1$ has a positive density on $(0, L)$ for some $0<L<\infty$. Then there is a positive invariant probability measure $\mu$ and the distribution of $X_t$ converges to $\mu$ whenever $X_0=x>0$.
\end{thm}
We note that the above theorem provides a classification of the stochastic Ricker model if the random variable $\xi_1$ has a density and is supported on $(0,L)$ for some $L>0$. One can also fully classify (\cite{FH02}) the stochastic Ricker model if the random coefficients do not have compact support.
\begin{thm}\label{t:pers_single3}
Consider the stochastic Ricker model $X_{t+1} = X_t \exp(r_{t+1}-a_{t+1}X_t)$ where
\begin{itemize}
  \item $r_1,\dots$ is a sequence of i.i.d random variables such that $\E[r_1]<\infty$ and $r_1$ has positive density on $(-\infty,+\infty)$,
  \item $a_1,\dots$ is a sequence of positive i.i.d random variables independent of $r_t$ such that $\E[a_1]<\infty$ and
  \item there exists $x_c$ such that $\E[\exp(r_1x)]<\infty$ for all $x\in [0, x_c]$.
\end{itemize}
Then if $\E[r_1]<0$, $X_t\to 0$ with probability 1 while if $\E[r_1]>0$ there is a positive invariant measure $\mu$ such that $X_t$ converges weakly to $\mu$.
\end{thm}

\subsection{Two species systems}
Suppose we have a two species system. The following result appeared in work by \cite{E89}.
\begin{thm}
Assume the following assumptions are satisfied
\begin{itemize}
\item For each $i=1,2$ there exists a positive invariant measure $\mu_i$ such that the distribution $\mu_i$ such that the distribution of $X_t^i$ converges to $\mu_i$ weakly whenever $X_0^i>0$ and $X_0^j=0$.
\item The mean per capita growth rates $r_i(\bx)$ are continuous functions.
  \item The process $\BX$ is irreducible on $(0,\infty)\times (0,\infty)$,
  \item For any Borel measurable $A\subset \R_+^2$ we have $\PP(\BX_1\in A|\BX_0=\bx_n) \to \PP(\BX_1\in A|\BX_0=\bx)$ whenever $\bx_n\to \bx$.
  \item For any $\bx\in \R_+^2$, $\sup_{t>0}\E[\ln^+ X_t^i~|~\BX_0=\bx]<\infty$ for $i=1,2$.
\end{itemize}
If $r_1(\mu_2)>0$ and $r_2(\mu_1)>0$ then there exists a unique positive invariant measure $\mu$ and the distribution of $\BX_t$ converges to $\mu$ weakly whenever $X_0^1, X_0^2>0$.
\end{thm}

\subsection{General criteria for coexistence}
Assume we have a general $n$ species system modeled by
\begin{equation}\label{e:SDE_gen}
X^i_{t+1} = X_t^i f_i(\BX_t, \xi_{t+1}), i=1,\dots,n.
\end{equation}
The subset $\Se\subset \R_+^n$ will denote the state space of the dynamics. It will either be a compact subset of $\R_+^n$ or all of $\R_+^n$.
The coexistence set is the subset $\Se_+=\{\bx\in \Se~|~x_i>0, i=1,\dots n\}$ of the state space where no species is extinct.
We will make the following assumptions:
\begin{itemize}
\item[\textbf{(A1)}] $\xi_1,\dots,\xi_n,\dots$ is a sequence of i.i.d. random variables taking values in a Polish space $E$.
\item[\textbf{(A2)}] For each $i$ the fitness function $f_i(\bx,\xi)$ is continuous in $\bx$, measurable in $(\bx,\xi)$ and strictly positive.
\item[\textbf{(A3)}] If the dynamics is unbounded: There exists a function $V:\Se_+\to\R_+$ and constants $\gamma_1,\gamma_3, C>0$ and $\rho\in(0,1)$ such that for all $\bx\in \Se_+$ we have
\begin{itemize}
\item[(i)] $V(\bx)\geq |\bx|^{\gamma_1}+1$,
\item[(ii)] $\E \left[V(\bx^\top f(\bx,\xi_1))\ell(\bx,\xi_1)\right]\leq \rho V(\bx)+C$,
where $$\ell(\bx,\xi):=\left(\max_{i=1}^n \left\{\max \left\{f_i(\bx,\xi), \frac{1}{f_i(\bx,\xi)}\right\}\right\}\right)^{\gamma_3}.$$
\end{itemize}
\item [\textbf{(A4)}] If the dynamics is bounded: There exists a constant $\gamma_3>0$ such that for all $\bx\in \Se_+$ we have
\[
\E \left[\ell(\bx,\xi_1)\right]<\infty.
\]
\end{itemize}

\begin{rmk}
In particular, if one supposes the conditions
\begin{itemize}
  \item [1)] There is a compact subset $K\subset \R_+^n\times \R^{\kappa_0}$ such that all solutions $\BX_t$ satisfy $\BX_t\in K$ for $t\in \Z_+$ sufficiently large;
  \item [2)] For all $i=1,2,\dots,n $, $$\sup_{\bx,\xi}|\ln f_i(\bx,\xi)|<\infty;$$
\end{itemize}
then assumption (A4) is satisfied.
\end{rmk}

Assumptions (A1) and (A2) ensure that $\BX_t$ is a Feller process that lives on $\Se_{+}$, i.e. $\BX_t\in\Se_+, t\in \Z_+$ whenever $\BX_0\in\Se_{+}$. One has to make the extra assumptions (A3) or (A4) in order to ensure the process does not blow up or fluctuate too abruptly between $0$ and $\infty$. We note that most ecological models will satisfy these assumptions. For more details see the work by \cite{BS19, CHN19}.

We will follow the notation, methods and results developed by \cite{MT, B18}. A point $\by\in \R_+^n$ is said to be \textit{accessible} from $\bx\in\Se_+$ if for every neighborhood $U$ of $y$ there exists $t\geq 0$ such that $P_t(\bx,U)>0$. Define $$\Gamma_\bx:=\{\by\in\Se_+ ~|~\by~\text{is accessible from}~\bx\}$$ and for $A\subset \R_+^n$
\[
\Gamma_A=\bigcap_{\bx\in A}\Gamma_\bx.
\]
Note that $\Gamma_A$ is the set of points which are accessible from every point of $A$. We say a set $A$ is \textit{accessible} if for all $\bx\in \R_+^{n,\circ}$
\[
\Gamma_{\bx}\cap A\neq \emptyset.
\]
Suppose there exist $\bx^*\in \Gamma_{\Se_+}$, a neighborhood $U$ of $\bx^*$, and a non-zero measure $\phi$ on $\Se_+$, such that for any $\bx\in U$
there is $m^*\in \mathbb{Z}_+$ such that
\[
\PP_ \bx(\BX_{m^*}\in \cdot)\geq \phi(\cdot).
\]
We will assume that such conditions are satisfied in our models. In many cases it is not hard to check that these conditions hold - see \cite{E89, CHN19}
Suppose the dynamics happens in either a compact subset of $\R_+^n$ or in $\R_+^n$. We denote the state space of the dynamics by $\Se$. We define the extinction set, where at least one species is extinct, by
\[
\Se_0:=\{\bx\in\Se~:~\min_i x_i=0\}.
\]
For any $\eta>0$ let
\[
\Se_\eta:=\{\bx\in\Se~:~\min_i x_i\leq \eta\}
\]
be the subset of $\Se$ where at least one species is within $\eta$ of extinction. Denote by $\M$ the set of all ergodic invariant probability measures supported on $\Se_0$ and by $\Conv(\M)$ the set of all invariant probability measures supported on $\Se_0$.
We say \eqref{e:SDE_gen} is \textit{stochastically persistent in probability} (\cite{C82}) if for all $\eps>0$ there exists $\eta(\eps)=\eta>0$ such that for all $\bx\in \Se_+$
\[
\liminf_{t\to\infty} \PP_{\bx}\{\BX_t\notin S_\eta\}>1-\eps.
\]
For any $t\in \N$ define the \textit{normalized occupation measure}
\[
\Pi_t(B):=\frac{1}{t}\sum_{s=1}^t\delta_{\BX(s)}(B)
\]
where $\delta_{\BX(s)}$ is the Dirac measure at $\BX(s)$ and $B$ is any Borel subset of $\Se$. Note that $\Pi_t$ is a random probability measure and $\Pi_t(B)$ tells us the proportion of time the system spends in $B$ up to time $t$.
Denote the (random) set of weak$^*$-limit points of $(\Pi_t)_{t\in \N}$ by $\U=\U(\omega)$. We say \eqref{e:SDE_gen} is \textit{almost surely stochastically persistent} (\cite{S12, BS19}) if for all $\eps>0$ there exists $\eta(\eps)=\eta>0$ such that for all $\bx\in \Se_+$
\[
\liminf_{t\to\infty} \Pi_t(\Se\setminus \Se_\eta)>1-\eps, ~\BX(0)=\bx.
\]

The following general theorem (whose full proof will appear in \cite{CHN19}) gives us persistence for a general $n$ species system.
\begin{thm}
Suppose that for all $\mu\in\Conv(\M)$ we have
\begin{equation}\label{e:pers_gen}
\max_{i} r_i(\mu)>0.
\end{equation}
Then the system is almost surely stochastically persistent and stochastically persistent in probability. Under additional irreducibility conditions, there exists a unique invariant probability measure $\pi$ on $\Se_+$ and as $t\to\infty$ the distribution of $\BX_t$ converges in total variation to $\pi$ whenever $\BX(0)=\bx\in \Se_+$. Furthermore, if $w:\Se_+\to\R$ is continuous and either bounded or satisfies
\[
w(\bx)\leq \E \left[V(\bx^\top f(\bx,\xi_1))\ell(\bx,\xi_1)\right], \bx\in \Se_+
\]
then $$\E w(\BX_t) \to \int_{\Se_+} w(\bx)\,\pi(d\bx).$$
\end{thm}
\begin{proof}[Sketch of proof]

First, using the Markov property and Assumption A3) one can show that for all $t\in \Z_+$ and $\bx\in \Se$
$$\E_\bx(V(\BX_t)\leq \rho^t V(\bx)+\dfrac{C}{1-\rho},$$
and
$$
\begin{aligned}
\E_\bx \ell(\BX_t, \xi_{t+1})
\leq&\rho^{t+1} V(\bx)+\dfrac{C}{1-\rho}.
\end{aligned}
$$

As a next step, one can show that if a continuous function $w$ satisfies $\lim_{\bx\to\infty} \frac{w(\bx)}{\E [V(\bx^T \mathbf{f}(\bx,\xi_t))\ell(\bx,\xi_t)]}=0$
then $w$ is $\mu$-integrable for any invariant probability measure $\mu$ of $\BX$.
Moreover, the strong law of large numbers for martingales will show that
 \begin{equation}\label{e5-B3}
 \lim_{T\to\infty}
\frac1T \sum_0^T\left(\log f_i(\BX_{t+1})-P \log f_i(\BX_t)\right)=0, \,\text{when}~ \BX(0)=\bx
\end{equation}
where $P$ is the transition operator of $\BX_t$. This combined with arguments by \cite{BS19} implies that if $\mu(\Se_+)=1$ then
$r_i(\mu)=0$ for any $i\in I$.

The next step is to show that there exist $M, C_2, \gamma_4>0,\rho_2\in(0,1)$ such that
$$
\E_\bx\left[ V(\BX_1)\prod_{i=1}^n X_i^{p_i}(1)\right]\leq \left(\1_{\{|\bx|<M\}}(C_2-\rho_2)+\rho_2\right) V(\bx)\prod_{i=1}^nx_i^{p_i}, ~\bx\in \Se
$$
for any $\bp=(p_1,\dots,p_n)\in\R^n$ satisfying
\begin{equation}\label{e:p}
|\bp|_1:=\sum|p_i|\leq\gamma_4.
\end{equation}
It is shown in \cite{SBA11} by the min-max principle that Assumption \eqref{e:pers_gen} is equivalent to the existence of $\mathbf p>0$ such that
\begin{equation}\label{e.p}
\min\limits_{\mu\in\M}\left\{\sum_{i}p_i r_i(\mu)\right\}:=2r^*>0.
\end{equation}

One can then prove, using arguments by \cite{HN18}, that there exists an integer $T^*>0$ such that, for any $T>T^*$, $\bx\in\Se_0, |\bx|<M$ one has
\begin{equation}\label{lm3.1-e0}
\sum_{t=0}^T\E_\bx\left(\ln V(\BX_{t+1})-\ln V(\BX_t)-\sum p_i\ln f_i(\BX_t, \xi_{t+1})\right)\leq-r^*(T+1).
\end{equation}
Define $U:\Se_+\to\R_+$ by $$U(\bx)=V(\bx)\prod_{i=1}^nx_{i}^{-p_i}$$ with $\bp$ and $r^*$ satisfying \eqref{e.p}. Let $n^*\in\N$ be such that
\begin{equation}\label{e:n*}
\rho_2^{1-n^*}>C_2.
\end{equation}
Using the previous results, as well as the analysis developed by \cite{B18, HN18} one can prove the following:
There exist numbers $\theta\in\left(0,\frac{\gamma_4}2\right)$, $K_\theta>0$, such that for any $T\in[T^*,n^*T^*]\cap \Z$ and $\bx\in\Se_+, \|\bx\|\leq M$,
$$\E_\bx U^\theta(\BX_T)\leq U^\theta(\bx)\exp\left(-\frac{1}{2}\theta r^*T\right) +K_\theta.$$

One can show that the process $(\rho^{-t}_2 U(\BX(t)))_{t\geq 0}$ is a supermartingale and use this in conjunction with the Markov property to show that there exist numbers $\kappa=\kappa(\theta,T^*)\in(0,1)$ and $\tilde K=\tilde K(\theta,T^*)>0$ such that
\begin{equation}\label{e:lya}
\E_\bx  U^\theta(\BX_{n^*T^*})\leq \kappa U^\theta(\bx)+\tilde K\,\text{ for all }\, \bx\in\Se_+.
\end{equation}
If the Markov chain $\BX_t$ is irreducible and aperiodic on $\Se_+$,
and a compact set is petite, then one can use the well known results by \cite{MT} in conjunction with the Lyapunov condition \eqref{e:lya} to conclude that there is $c_4>1$ such that for all $\bx\in \Se_+$
$$ c_4^t\|P_t(\bx,\cdot)-\pi(\cdot)\|_{TV}\to 0\text{ as } t\to\infty,$$
where $\|\cdot\|_{TV}$ is the total variation distance. In particular, this implies that the distribution of $\BX_t$ converges weakly to $\pi$ as $t\to\infty$.
\end{proof}

\section{Two species systems}\label{s:a2d} In general, one needs stronger assumptions for extinction. We will assume for simplicity $n\leq 2$, so that we have one or two species.
We need one more condition for extinction. This condition makes sure that the martingale part of $\BX_t$ is bounded and that the family of occupation measures $(\Pi_t)_{t\in \Z_+}$ is tight.

\textbf{A5)}
 There exists a function $\phi:\Se\to(0,\infty)$ and constants $C, \delta_\phi>0$ such that for all $\bx\in\Se$
\[
\E_\bx V(\BX_1)\leq V(\bx)-\phi(\bx)+C
\]
and
\[
\E_{\bx}\left(V(\BX_1)-\E_\bx V(\BX_1)\right)^2+\E\left|\log f(\bx,\xi_1)-\E\log f(\bx,\xi_1)\right|^2\leq \delta_\phi \phi(\bx).
\]

Define $\Se^j:=\{\bx\in\Se~|~x_i=0, i\neq j\}$ to be the subspace supported by the species $j$. If we restrict the process to $\Se^j$ then the extinction set is given by $\Se_0:=\{0\}$ and the persistence set by $\Se_+^j:=\Se^j\setminus \{0\}$. Let $\M^j:=\{\mu\in\M~|~\mu(\Se^j)=1\}, \M^{j,+}:=\{\mu\in\M~|~\mu(\Se_+^j)=1\}$ be the sets of ergodic probability measures
on $\Se^j$ and $\Se^{j}_+$. We also assume that the subspaces $\Se_0^1, \Se_0^2, \Se_+$ are accessible, i.e., we can get close to them from any starting point $\bx\in\Se_+$ with positive probability, and each subspace supports at most one ergodic probability measure. Consider two species interacting via the general system
\begin{equation}\label{e:2d_disc}
\begin{aligned}
X^1_{t+1}&=X^1_t f_1(X^1_t, X^2_t,\xi_{t+1}),\\
X^2_{t+1}&=X^2_t f_2(X^1_t, X^2_t,\xi_{t+1}),
\end{aligned}
\end{equation}
The results by \cite{CE89, E89} assumed some type of monotonicity and only looked at competitive behavior. They can be generalized  as follows (see \cite{CHN19} for proofs). We first look at the Dirac delta measure $\delta_0$ at the origin $(0,0)$
\[
r_i(\delta_0) = \E[\ln f_i(0,\xi_1)], i=1,2.
\]
If $r_i(\delta_0)>0$ then species $i$ survives on its own and converges to a unique invariant probability measure $\mu_i$ supported on $\Se_+^i := \{\bx\in\Se~|~x_i\neq 0, x_j=0, i\neq j\}$. Remember that the (random) set of weak$^*$-limit points of the family of occupation measures $(\Pi_t)_{t\in \N}$ is denoted by $\U=\U(\omega)$. Thus, if we say that $\U(\omega)=\{\mu_1\}$, this means that for the realization $\omega$ we have $\Pi_t \to \mu_1$ weakly.
\begin{enumerate}[label=(\roman*)]
  \item Suppose $r_1(\delta_0)>0, r_2(\delta_0)>0$. The realized per-capita growth rates can be computed via
  \[
  r_i(\mu_j)=\int_{(0,\infty)}\E[\ln f_i(x,\xi_1)]\mu_j(dx).
  \]
  \begin{itemize}
    \item If $r_1(\mu_2)>0$ and $r_2(\mu_1)>0$ we have coexistence and convergence of the distribution of $\BX_t$ to the unique invariant probability measure $\pi$ on $\Se_+$.
    \item If $r_1(\mu_2)>0$ and $r_2(\mu_1)<0$ we have the persistence of $X^1$ and extinction of $X^2$. In other words, for any $\bx\in\Se_+$
    \[
    \PP_\bx\left\{\U(\omega)=\{\mu_1\} ~\text{and}~\lim_{t\to\infty}\frac{\ln X^2_t}{t}=r_2(\mu_1)<0, \right\}=1.
    \]
    \item If $r_1(\mu_2)<0$ and $r_2(\mu_1)>0$ we have the persistence of $X^2$ and extinction of $X^1$. In other words, for any $\bx\in\Se_+$
    \[
    \PP_\bx\left\{\U(\omega)=\{\mu_2\} ~\text{and}~\lim_{t\to\infty}\frac{\ln X^1_t}{t}=r_1(\mu_2)<0, \right\}=1.
    \]
    \item If $r_1(\mu_2)<0$ and $r_2(\mu_1)<0$ we have that for any $\bx\in\Se_+$
    \[
    p_{\bx,j}:=\PP_\bx\left\{\U(\omega)=\{\mu_j\} ~\text{and}~\lim_{t\to\infty}\frac{\ln X^i_t}{t}=r_i(\mu_j)<0, i\neq j \right\}
    \]
    and
    \[
    p_{\bx,1}+ p_{\bx,2}=1.
    \]
  \end{itemize}
  \item Suppose $r_1(\delta_0)>0, r_2(\delta_0)<0$. Then specis $1$ survives on its own and converges to its unique invariant probability measure $\mu_1$ on $\Se^1_+$.
  \begin{itemize}
    \item If $r_2(\mu_1)>0$ we have the persistence of both species and convergence of the distribution of $\BX_t$ to the unique invariant probability measure $\pi$ on $\Se_+$.
    \item If $r_2(\mu_1)<0$ we have the persistence of $X^1$ and the extinction of $X^2$. In other words, for any $\bx\in\Se_+$
    \[
    \PP_\bx\left\{\U(\omega)=\{\mu_1\} ~\text{and}~\lim_{t\to\infty}\frac{\ln X^2_t}{t}=r_2(\mu_1)<0, \right\}=1.
    \]
  \end{itemize}
\item Suppose $r_1(\delta_0)<0, r_2(\delta_0)<0$. Then both species go extinct with probability one. In other words, for any $\bx\in\Se_+$
\[
\PP_\bx\left\{\lim_{t\to\infty}\frac{\ln X^i_t}{t}=r_i(\delta_0)<0 \right\}, i=1,2.
\]
\end{enumerate}
We note that our results are significantly more general than those from \cite{E89}. In \cite{E89} the author only gives conditions for coexistence, and does not treat the possibility of the extinction of one or both species.
\end{document}